%% file: root_Arxiv.tex
\documentclass[letterpaper, 10 pt, conference]{ieeeconf}  % Comment this line out if you need a4paper

\IEEEoverridecommandlockouts                              % This command is only needed if 
\usepackage{enumerate}
\usepackage{mathtools}
\usepackage{subfigure}
\usepackage{amssymb}

\usepackage{amsthm}

\input{preamble}

\renewcommand{\baselinestretch}{0.93}
\renewenvironment{quote}{\list{}{\leftmargin=0.in\rightmargin=0.in}\item[]}{\endlist}

\newcommand{\mpmargin}[2]{#1}

\title{\LARGE \bf
A Queueing Network Approach to the Analysis and Control of Mobility-On-Demand Systems \vspace{-0.5em}
}

\author{Rick Zhang and Marco Pavone% <-this % stops a space
\thanks{Rick Zhang and Marco Pavone are with the Department of Aeronautics \& Astronautics, Stanford University, Stanford, CA 94305 {\tt\small \{rickz, pavone\}@stanford.edu}}
}

\begin{document}

\maketitle
\thispagestyle{empty}
\pagestyle{empty}

%%%%%%%%%%%%%%%%%%%%%%%%%%%%%%%%%%%%%%%%%%%%%%%%%%%%%%%%%%%%%%%%%%%%%%%%%%%%%%%%
\begin{abstract}
This paper presents a queueing network approach to the analysis and control of mobility-on-demand (MoD) systems for urban personal transportation. A MoD system consists of a fleet of vehicles providing one-way car sharing service and a team of drivers to rebalance such vehicles. The drivers then rebalance themselves by driving select customers similar to a taxi service. We model the MoD system as two coupled closed Jackson networks with passenger loss. We show that the system can be approximately balanced by solving two decoupled linear programs and exactly balanced through nonlinear optimization. The rebalancing techniques are applied to a system sizing example using taxi data in three neighborhoods of Manhattan, which suggests that the optimal vehicle-to-driver ratio in a MoD system is between 3 and 5. Lastly, we formulate a real-time closed-loop rebalancing policy for drivers and demonstrate its stability (in terms of customer wait times) for typical system loads.

\end{abstract}

%%%%%%%%%%%%%%%%%%%%%%%%%%%%%%%%%%%%%%%%%%%%%%%%%%%%%%%%%%%%%%%%%%%%%%%%%%%%%%%%
\section{Introduction}

%Takeaway messages: 
%1) Jackson network approach is useful for designing a MoD system with human rebalancers and can effectively evaluate the trade-off between the number of vehicles and the number of drivers.
%2) human driven MoD systems cannot offer the same performance as AMoD systems because not all vehicles are able to rebalance
%3) The insights gained from the queueing network modeling were used to design a real-time rebalancing policy for drivers

% graph showing 

%This report analyzes the problem of rebalancing a mobility-on-demand system with human rebalancers from a queueing-theoretical perspective by modeling the system using closed Jackson network theory. In \cite{Smith2013ACC}, calculations performed using a fluidic model suggested that the ratio of minimum drivers to vehicles is around 1/3 to 1/4. However, the result does not address any performance qualifications about what kind of quality of service can be achieved with these ratios. In this report, we provide such performance measures (namely, the availability of vehicles at each station) using a queueing network approach.
%Private automobiles, the primary mode of personal mobility in the past century, has become impractical in many large urban centers around the world due to limited urban space available for roads and parking. The resulting congestion and heavy costs for car owners have prompted research into alternative modes of personal mobility urban cities. 

Car sharing promises to be a cost effective and sustainable alternative to private urban mobility by allowing a split of hefty ownership costs, increasing vehicle utilization, and reducing urban infrastructure needed for parking \cite{WJM-CEBB-LDB:10}. One type of vehicle-sharing service, called mobility-on-demand (MoD), consists of stacks or racks of light electric vehicles parked at many different stations throughout a city \cite{WJM-CEBB-LDB:10}. Each customer arrives at a station, takes a vehicle to the desired destination, and drops off the vehicle at that station. Due to the asymmetry of customer demands, vehicles will aggregate at some stations and be depleted elsewhere, causing the system to become unbalanced \cite{MP-SLS-EF-DR:12}. Rebalancing the system has been studied in \cite{MP-SLS-EF-DR:12, RZ-MP:14, JC-KHL-CKYT:13} for MoD systems with self-driving cars and in \cite{SLS-MP-MS-ea:13} for human-driven MoD systems under a fluidic model. To rebalance the MoD system in the absence of self-driving cars, the strategy is to hire human drivers to drive excess vehicles to stations where they are needed. The drivers then themselves are ``rebalanced'' by driving select customers to their destinations as a taxi service. In this way, the MoD system can be viewed as an one-way customer-driven car sharing service mixed with a taxi service. 
%We note in practice the drivers may also redistribute themselves via public transit or riding with other drivers, but for simplicity and comparison purposes these modes of comparison are not addressed in this paper.

%\emph{Statement of contributions:} 
The objective of this paper is to develop a queueing network framework for the analysis and control  of (human-driven) mobility-on-demand systems. We then apply the insights from this queueing framework to develop real-time closed-loop policies to control such systems. On the modeling and analysis side, we consider a model similar to the one proposed in \cite{SLS-MP-MS-ea:13}, which hinges upon the optimization of rebalancing \emph{rates} and is studied under a fluidic approximation (where customers, drivers, and vehicles are modeled as a continuum). The model in \cite{SLS-MP-MS-ea:13} offers insights into the minimum number of vehicles and drivers required in a MoD system but does not provide key performance metrics in terms of quality of service (i.e., the availability of vehicles at stations or the customer wait times). These shortcomings are addressed by \cite{RZ-MP:14} for an autonomous MoD system, where the system is modeled as a stochastic queueing network from which key performance metrics are derived. This paper can be viewed as an extension of the models in \cite{SLS-MP-MS-ea:13, RZ-MP:14} to human-driven MoD systems taking into account both vehicles and rebalancing drivers. On the control side, real-time closed-loop policies for one-way car sharing systems have been studied in \cite{DJ-GHAC-CB:14} and \cite{RN-EMH:11} with the objective of maximizing profit, where the rebalancing of vehicles is modeled as a cost. Our paper differs from these works in two key respects: 1) in addition to minimizing cost, our key objective is quality of service for customers in terms of vehicle availabilities and wait times, and 2) we explicitly control the movement of rebalancing drivers which makes the system self-contained (e.g., drivers do not need to rely on public transit to rebalance themselves).

Our contribution in this paper is fourfold. First, we model a MoD system within a queueing network framework that takes into account the \emph{coupled} rebalancing of vehicles and drivers. Specifically, our approach is to model a MoD system as two coupled closed Jackson networks with passenger loss.
Second,  we present two approaches for the open-loop control of a MoD system. In the first approach, the optimal rebalancing parameters are solved by two decoupled linear programs, and are therefore efficient to compute, but only approximately guarantee balance of the system. In the second approach, nonlinear optimization techniques are used (with higher computational cost) to balance the system exactly. 
 Third, we apply such approaches  to the problem of system sizing. Our key finding is that the optimal vehicle-to-driver ratio in a MoD system should be between 3 and 5.  Finally, leveraging the aforementioned open-loop control strategies, we devise a real-time \emph{closed-loop} rebalancing policy and demonstrate its stability under typical system loads.

The rest of this paper proceeds as follows: Section \ref{sec:back} reviews some key results in the theory of Jackson networks. Section \ref{sec:model} describes in detail our queueing network model of a MoD system. Section \ref{sec:rebalancing} offers the approaches for the open-loop control of a MoD system.
% In the first approach, the optimal rebalancing parameters are solved by two decoupled linear programs, and are therefore efficient to compute and easy to scale, but only approximately balance the system. In the second approach, nonlinear optimization techniques are used (with higher computational cost) to balance the system exactly. 
 The rebalancing techniques are then applied to a system sizing example based on taxi data in Manhattan. In Section \ref{sec:realTime} we introduce a real-time closed-loop control policy useful for practical systems. Finally, in Section \ref{sec:conc} we draw our conclusions and provide directions for future research.

%In the past few years, car sharing services have emerged as a strong alternative to private vehicle ownership in densely populated cities 

\section{Background Material}\label{sec:back}

In this \mpmargin{section}{If you have time later, think a little bit more about probability distribution} we review several useful results and techniques from the theory of queueing networks, in particular the theory of Jackson networks. We consider a directed graph $G(V,E)$ where the set of vertices $V$ represent first-in-first-out service nodes or queues. Discrete agents (often referred to as customers in the literature) arrive from outside the network according to a stochastic process and are serviced at each node. The agents then travel to other nodes in the network or leave the system. A network in which a fixed number of agents move among the nodes with no external arrivals or departures is referred to as a \emph{closed} network (in contrast, agents in \emph{open} networks arrive externally and eventually depart from the network). A Jackson network is a class of Markovian queueing networks whereby  the routing distribution, $r_{ij}$, is stationary  and the service rate at each node $i$, $\mu_{i}(x_i)$, only depends on the number of agents at that node, $x_i$ \cite[p. 9]{RS:99}. Jackson networks are part of a broader class of networks called BCMP networks \cite{FB-KMC-RRM-FGP:75} that are known to admit product-form stationary distributions, making them relatively easy to analyze. The stationary  distribution $\{\pi_i\}_{i=1}^{|V|}$ of the underlying Markov chain of a Jackson network satisfies the balance equations
\begin{equation}
\pi_i = \sum_{j \in V} \pi_j r_{ji}, \;\;\;\; \text{for all } i \in V. 
\label{eq:traffic}
\end{equation}
For a closed network, \eqref{eq:traffic} does not yield a unique solution and only determines $\pi = (\pi_1\ \pi_2\ ...\ \pi_{\lvert V \rvert})^T$ up to a constant factor. Accordingly, for a closed network $\pi$ is referred to as the \emph{relative} throughput. The stationary probability distribution of a closed Jackson network with $m$ agents is given by  
\begin{equation*}
\mathbb{P}(x_1, x_2, ..., x_{\lvert V \rvert}) = \frac{1}{G(m)} \prod_{j = 1}^{\lvert V \rvert} \pi^{x_j}_j \prod_{n=1}^{x_j} \mu_j(n)^{-1},
\end{equation*}
where $G(m)$ is a normalization constant required to make $\mathbb{P}(x_1, x_2, ..., x_{\lvert V \rvert})$ a probability measure. It turns out that many performance metrics of the network can be expressed in terms of the normalization constant $G(m)$. Two such performance metrics are of interest to us: 1) the actual throughput of each node (see \cite[p. 27]{RS:99}) is given by
\begin{equation}
\Lambda_i(m) = \pi_i \, G(m-1)/G(m),
\label{eq:throughput}
\end{equation}
and 2) the probability that a node has at least one agent, referred to as the \emph{availability} of node $i$ (\cite{RZ-MP:14, DKG-CHX:11}), is given by
\begin{equation}
A_i(m) = \gamma_i \, G(m-1)/G(m),
\label{eq:Ai}
\end{equation}
where $\gamma_i = \pi_i / \mu_i(1)$ is referred to as the relative utilization of node $i$. 

% new stuff
In general, solving for $G(m)$ is quite computationally expensive, especially when $m$ is large. A well-known technique called mean value analysis (MVA) \cite{MR-SSL:80} enables us to compute the mean values of the performance metrics without explicitly solving for $G(m)$. MVA computes the mean wait times and queue lengths at each node of the Jackson network at each iteration. The algorithm begins by assuming the system contains a single agent and adds one agent to the system at each iteration until the desired number of agents is reached. The MVA algorithm is described in detail in \cite{RZ-MP:14, SSL:83}, and is used extensively in this paper to compute performance metrics formally introduced in Section \ref{sec:formulation}.

\section{Model Description and Problem Formulation} \label{sec:model}
\subsection{MoD system model}
In this section we formally describe the MoD system under consideration and cast it within a queueing network framework by modeling the system as two coupled, closed Jackson networks.
We consider $N$ stations with unlimited parking capacity placed in a given geographical area, $m_v$ vehicles that can be rented by customers for one-way trips between stations, and $m_d$ ``rebalancing'' drivers employed to rebalance the vehicles by driving them to the stations where they are needed. The drivers then ``rebalance'' themselves by driving customers to their destinations, operating as a taxi service. These assumptions require each driver to \emph{always} have access to a vehicle since the driver's task involves driving a vehicle with or without a customer (A driver left at a station without a vehicle is effectively ``stranded''). We therefore pose the constraint $v_i \geq d_i$, where $v_i$ is the number of vehicles at station $i$ and $d_i$ is the number of drivers at station $i$. With this requirement, we may view the MoD system as two systems operating in parallel -- a one-way customer-driven car sharing service with $m_v - m_d$ vehicles and a taxi service with $m_d$ vehicles. It is worth noting that there are other, more elaborated ways of managing a MoD system which we do not address in this paper. For example, in \cite{SLS-MP-MS-ea:13}, the authors also consider customers potentially riding with multiple drivers. One could also envision a system where drivers can drive other drivers or take public transportation to stations with excess cars. The extension of our model to such cases is an interesting avenue for future research.
\begin{figure}[htbp]
	\centering \vspace{-1.3em}
		\subfigure{\label{fig:mod} \includegraphics[width=0.31\textwidth]{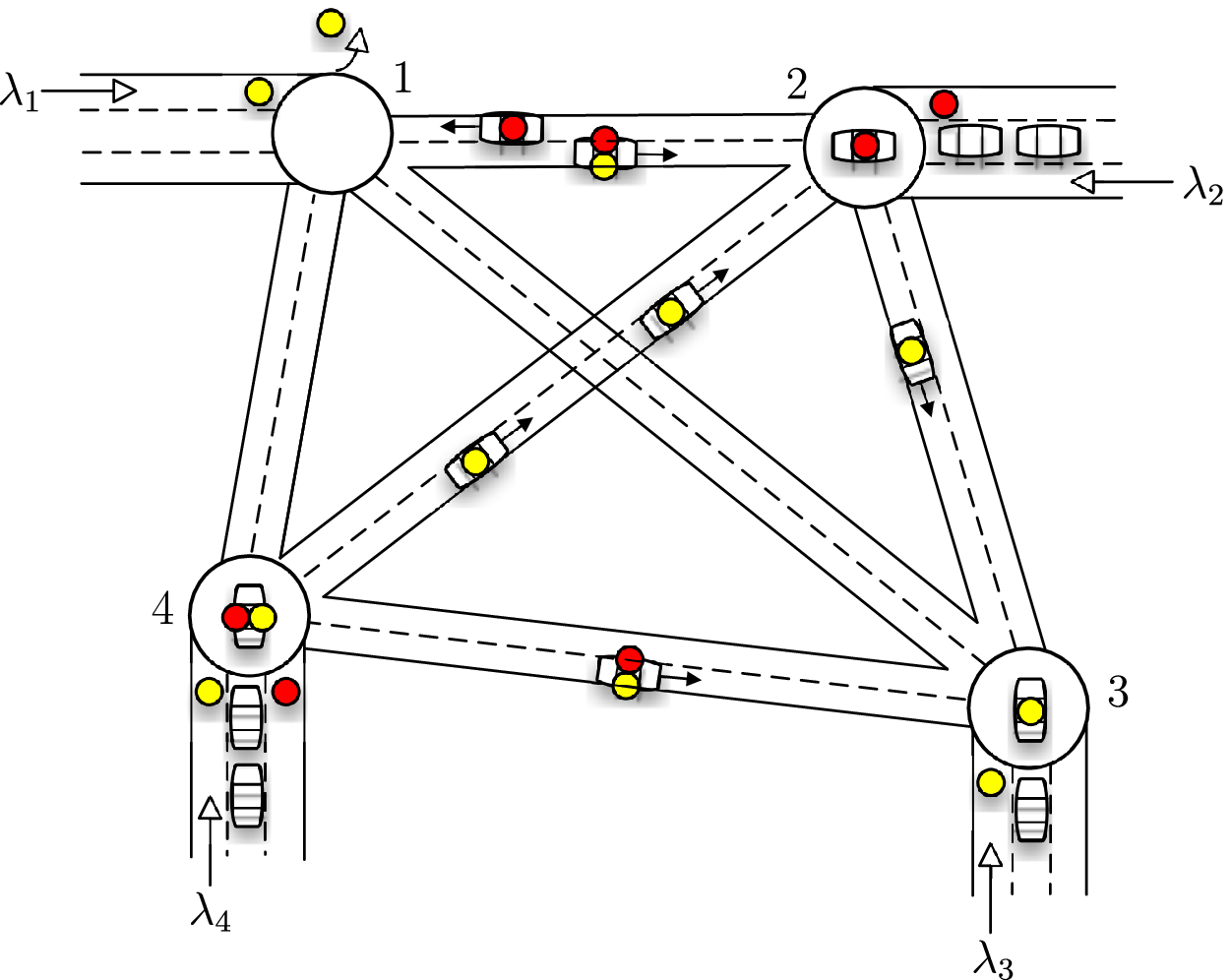}} \hspace{-1em}
		\subfigure{\label{fig:station} \includegraphics[width=0.11\textwidth]{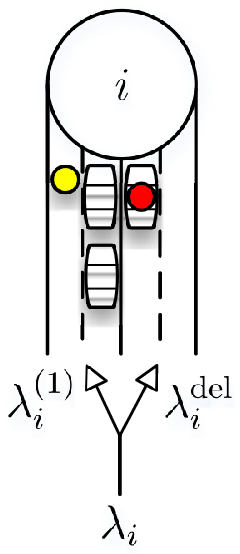}} \vspace{-0.5em}
	\caption{Left: MoD system model. Yellow dots represent customers and red dots represent rebalancing drivers. Customers can drive themselves or ride with a rebalancing driver. Customers are lost if no vehicles are available (station 1). Right: each customer arriving at station $i$ is delegated to either System 1 (customer-driven vehicles) or System 2 (taxi system).}
\label{fig:modsys}
\end{figure}

Customers arrive to station $i$ according to a Poisson process with parameter $\lambda_i$. Upon arrival at station $i$, the customer selects a destination $j$ with probability $p_{ij}$, where $p_{ij} \geq 0$, $p_{ii} = 0$, and $\sum_j p_{ij} = 1$. Furthermore, we assume that the probabilities $\{p_{ij}\}_{ij}$ constitute an irreducible Markov chain. The customer can travel to his/her destination in one of two ways: 1) the customer drives a vehicle to his/her destination, or 2) the customer is taken to his/her destination by a rebalancing driver. The travel time from station $i$ to station $j$ is an exponentially distributed random variable with mean $T_{ij} > 0$. Justification of assumptions such as Poisson arrivals and exponential travel times is discussed in \cite{RZ-MP:14}. We employ a ``passenger loss'' model similar to \cite{RZ-MP:14, DKG-CHX:11}, where if a vehicle is not available upon the arrival of a customer, the customer immediately leaves the system. However, due to the additional complexity of our MoD model (a one-way car sharing service and a taxi service in parallel) the passenger loss assumption is more involved. We assume that upon arrival at a station, a customer is delegated to one of the two parallel systems by the MoD service operator (see Fig. \ref{fig:modsys}). The customer is lost if there are no available vehicles in \emph{the system to which he/she was delegated}. For example, if a customer is delegated to the taxi system and no taxis are immediately available, the customer \emph{cannot} switch over to the other system and drive himself/herself to the desired destination. The modeling consequences of this assumption will be further discussed in the next section. The performance criterion of interest in this case is the probability a customer will find an available vehicle (both empty vehicles and taxis) at each station. In Section \ref{sec:realTime} we will relax the passenger loss assumption and investigate the more realistic scenario where customers form a queue to wait for available vehicles. The performance of the system is then measured by customer wait times.

\subsection{Jackson network model of a MoD system} \label{sec:jmodel}
We now formally cast the MoD model described in the previous section within a queueing network framework. The key is to construct an abstract queueing network where the stations are modeled as single-server (SS) nodes and the roads as infinite-server (IS) nodes, as done in \cite{RZ-MP:14, DKG-CHX:11}. Vehicles form a queue at each SS node while waiting for customers and are ``serviced'' when a customer arrives. The vehicle then moves from the SS node to the IS node connecting the origin to the destination selected by the customer. After spending an exponentially distributed amount of time (with mean $T_{ij}$) in the IS node, the vehicle moves to the destination SS node. This setup is then a \emph{closed Jackson network with respect to the vehicles}. To capture the idea that the MoD system consists of two systems (customer-driven system and taxi system) operating in parallel, we model the MoD system as two coupled closed Jackson networks. More formally, let System 1 represent the Jackson network of $m_v - m_d$ customer-driven vehicles, and System 2 represent the network of $m_d$ taxis. Let $S^{(k)}$ represent the set of SS nodes and $I^{(k)}$ represent the set of IS nodes in the $k^{\text{th}}$ Jackson network, where $k = \{1,2\}$. For each network, each SS node is connected to every other SS node through an IS node. Thus, each network consists of $N + N(N-1) = N^2$ nodes (the IS node from station $i$ to itself is not represented since $p_{ii} = 0$). For each IS node $i \in I^{(k)}$, let Parent$(i)$ and Child$(i)$ be the origin and destination of $i$, respectively. The routing matrix $\{r^{(k)}_{ij}\}_{ij}$ in Jackson network $k$ can then be written as
\begin{equation*}
r^{(k)}_{ij}\! =\!
\begin{cases}
	p^{(k)}_{il} & \! i \in S^{(k)}, j \in I^{(k)}, i = \text{Parent}(j), l = \text{Child}(j), \\
	1 & \! i \in I^{(k)}, j \in S^{(k)}, j = \text{Child}(i), \\
	0 &\! \text{otherwise},
\end{cases}
\label{eq:rij}
\end{equation*}
where the first case is the movement from a SS node to an IS node and the second case is from an IS node to its unique destination SS node. The service times at each node are exponentially distributed with mean service rates 
\begin{equation*}
\mu^{(k)}_{i}(n) = 
\begin{cases}
	\lambda^{(k)}_{i} & \hspace{-0.3em}\text{if } i \in S^{(k)}, \\
	\frac{n}{T_{jl}}& \hspace{-0.3em}\text{if } i \in I^{(k)}, j = \text{Parent}(i), l = \text{Child}(i),
\end{cases}
\label{eq:muij}
\end{equation*}
where $n$ is the number of vehicles in the IS node. 
%We model the human-driven MOD system using 2 coupled closed Jackson networks with passenger loss that take into account both vehicle rebalancing and driver rebalancing with ``rebalancing promoting'' policies (rather than enforcing a rebalancing rate). In each Jackson network, stations are modeled as single-server (SS) first-come-first-serve nodes (or queues) with exponential service times and the routes between pairs of stations are modeled as infinite-server (IS) nodes with exponentially distributed service times. More formally, let $S^{(k)}$ be the set of SS nodes in the $k^{\text{th}}$ Jackson network and $I^{(k)}$ be the set of IS nodes in the $k^{\text{th}}$ Jackson network, $k = 1,2$. The Jackson networks model the system from the vehicles' perspective rather than the passengers. Vehicles form queues at the SS nodes and are ``serviced'' when customers pick them up. Since passengers arrive at the stations according to a Poisson process, the inter-arrival times are exponentially distributed. If the customer in either network does not find a free vehicle upon arrival (i.e. the queue of vehicles is empty), the customer immediately leaves the system. The first Jackson network models $m_v-m_d$ empty vehicles while the second Jackson network models $m_d$ vehicles with rebalancing drivers. 
With this formulation we have defined two closed Jackson networks of the same form as in \cite{RZ-MP:14}, amenable to analysis. 

We now return to the customer arrival process and the loss model assumption. Recall that customers arrive at station $i$ according to a Poisson process with rate $\lambda_i$. Upon arrival, and depending on the destination, a customer is first delegated to either System 1 or System 2. This can be seen as a Bernoulli splitting of the customer arrival process into two Poisson processes for each desired destination. Denote by $\lambda^{(1)}_i$ the total rate of customers delegated to System 1, by $p^{(1)}_{ij}$ the routing probabilities associated with System 1 ($p^{(1)}_{ij} \geq 0$, $p^{(1)}_{ii}=0$, $\sum_j p^{(1)}_{ij}=1$), by $\lambda_i^{\mathrm{del}}$ the total rate of customers delegated to System 2, and by $\eta_{ij}$ the routing probabilities associated with System 2. We have the relationship
\begin{equation*}
\lambda_i = \lambda^{(1)}_i + \lambda_i^{\mathrm{del}}
\end{equation*}
for each station $i$. 
We also define $q_i$ to be the total fraction of customers delegated to System 1 at station $i$, i.e.,  $q_i = \lambda^{(1)}_i/\lambda_i$. We can also write $1 - q_i = \lambda_i^{\mathrm{del}}/\lambda_i$. The routing probabilities are related by
\begin{align}
p_{ij} &= \mathbb{P}(i \rightarrow j \mid \text{System 1}) \, p^{(1)}_{ij} + \mathbb{P}(i \rightarrow j \mid \text{System 2})\,  \eta_{ij} \nonumber \\
&= q_i \, p^{(1)}_{ij} + (1-q_i)\,  \eta_{ij}.
\label{eq:prob1}
\end{align}
We can equivalently say that the Poisson rate of customers originating at station $i$ and headed for station $j$ is $\lambda_i \, p_{ij}$. The arrival rate of these customers to System 1 is then $\lambda^{(1)}_i  p^{(1)}_{ij}$ and the arrival rate to System 2 is $\lambda_i^{\mathrm{del}}  \eta_{ij}$. Thus relation \eqref{eq:prob1} can be rewritten as 
\begin{equation}
\lambda_i\,  p_{ij} = \lambda^{(1)}_i p^{(1)}_{ij} + \lambda_i^{\mathrm{del}}   \eta_{ij}.
\end{equation} 
If the delegation process is known (i.e., $\lambda_i^{\mathrm{del}}$ and $\eta_{ij}$ are known), the routing probabilities for System 1 can be solved by rearranging \eqref{eq:prob1} as
\begin{equation}
p^{(1)}_{ij} = \frac{1}{q_i}\, p_{ij} - \frac{1-q_i}{q_i}\, \eta_{ij}.
\label{eq:prob2}
\end{equation}
In Section \ref{sec:formulation} we will describe in detail how to solve for $\lambda_i^{\mathrm{del}}$ and $\eta_{ij}$. Arrival rates $\lambda^{(1)}_i$, routing probabilities $p^{(1)}_{ij}$, and mean travel times $T_{ij}$ fully describe the System 1 Jackson network.  

Now we consider the second Jackson network, System 2, which models the $m_d$ vehicles operating as a taxi service. This network must not only provide service to customers but also rebalance the MoD system to ensure quality of service. To incorporate the notion of vehicle rebalancing, we use the concept of ``virtual'' customers as in \cite{RZ-MP:14}. Virtual customers are generated at station $i$ according to a Poisson process with parameter $\psi_i$ and routing probabilities $\xi_{ij}$, independent from the real customer arrival process. Virtual customers are lost upon arrival if a taxi is not immediately available, just like real customers. In this way, virtual customers promote rebalancing while not enforcing a strict rebalancing rate, which is key to retaining tractability in the model. The overall customer arrival rate (real and virtual) at station $i$ for System 2 is 
\begin{equation*}
\lambda^{(2)}_i = \lambda_i^{\mathrm{del}} + \psi_i.
\end{equation*} 
With respect to the vehicles, $\lambda^{(2)}_i$ is the exponentially distributed service rate at SS node $i$. The routing probabilities for this network can be defined as
\begin{align}
p^{(2)}_{ij} &= \mathbb{P}(i \rightarrow j \mid \text{ virtual})\,  \xi_{ij} + \mathbb{P}(i \rightarrow j \mid \text{ real})\,  \eta_{ij} \nonumber \\
&= \frac{\psi_i}{\lambda^{(2)}_i}\, \xi_{ij} + \frac{\lambda_i^{\mathrm{del}}}{\lambda^{(2)}_i} \, \eta_{ij} \nonumber \\
&= p_i \, \xi_{ij} + (1- p_i)\,  \eta_{ij},
\end{align}
where $p_i = \frac{\psi_i}{\lambda^{(2)}_i}$, similar to the definition in \cite{RZ-MP:14}. 

To summarize our Jackson network model, customers arrive at station $i$ headed for station $j$ according to a Poisson process with rate $\lambda_i \, p_{ij}$. Upon arrival, each customer is delegated to one of two systems, the customer-driven system (System 1) or the taxi system (System 2). The probability of the customer (going from station $i$ to $j$) being delegated to System 1 is $\frac{\lambda^{(1)}_i p^{(1)}_{ij}}{\lambda_i \, p_{ij}}$ and the probability of the customer delegated to System 2 is $\frac{\lambda_i^{\mathrm{del}}  \eta_{ij}}{\lambda_i \, p_{ij}}$ (from \eqref{eq:prob2}). Once the customer has been delegated, if he/she finds the station empty of vehicles, the customer immediately leaves the system. Once delegated, a customer cannot switch from System 1 to System 2 or vice versa. We note that in the same way that $\psi_i$ represents the rebalancing-promoting rate of vehicles in the MoD system, $\lambda_i^{\mathrm{del}}$ represents the rebalancing-promoting rate of the drivers. Together, the parameters $\psi_i$, $\xi_{ij}$, $\lambda_i^{\mathrm{del}}$, and $\eta_{ij}$ constitute the open-loop controls for our model of a MoD system. The open-loop control problem is formalized and solved in Section \ref{sec:rebalancing}.

\subsection{Performance criteria} \label{sec:formulation}
Our task to control the MoD system involves optimizing the parameters $\lambda_i^{\mathrm{del}}$ (rebalancing the drivers) and $\psi_i$ (rebalancing the vehicles) as well as the routing probabilities $\eta_{ij}$ and $\xi_{ij}$. The key performance metric is the availability of vehicles (the probability that a customer will find an available vehicle), given by \eqref{eq:Ai}. In \cite{DKG-CHX:11} it was shown that for a closed Jackson network of the form described in the previous section, the availability satisfies $\lim_{m \rightarrow \infty} A_i(m) = \gamma_i / \gamma_S^{\text{max}}$, for all $i \in S$, where $\gamma_i$ is the relative utilization at node $i \in S$, $S$ is the set of station nodes, and $\gamma_S^{\text{max}}:= \text{max}_{i \in S} \, \gamma_i$. As the number of vehicles increases, the set of stations $B:=\{i \in S: \gamma_i = \gamma_S^{\text{max}}\}$ will have availability approaching one while all other stations will have availability strictly less than one. Thus, a natural notion of rebalancing, introduced in \cite{RZ-MP:14} for autonomous MoD systems, is to ensure that $A_i(m) = A_j(m)$ for all $i,j \in S$ (or equivalently $\gamma_i = \gamma_j$ for all $i, j \in S$, as implied by \eqref{eq:Ai}). The relative utilizations for each Jackson network are defined as follows 
\begin{align*}
\gamma^{(1)}_i &= \frac{\pi^{(1)}_i}{\mu^{(1)}_i} = \frac{\pi^{(1)}_i}{\lambda_i - \lambda_i^{\mathrm{del}}} \;\;\;\; \forall i \in S^{(1)}, \\
\gamma^{(2)}_i &= \frac{\pi^{(2)}_i}{\mu^{(2)}_i} = \frac{\pi^{(2)}_i}{\lambda_i^{\mathrm{del}} + \psi_i} \;\;\;\; \forall i \in S^{(2)},
\end{align*}
where $\pi_i^{(k)}, i \in S^{(k)}, k = \{1,2\}$ satisfies \eqref{eq:traffic}.
For autonomous MoD systems, the constraint $\gamma_i = \gamma_j$ embodies two features: 1) fairness, as characterized by equal availability across all stations, and 2) performance, since the availability at each station approaches 100\% as the number of vehicles increases. We will apply this constraint to both Jackson networks in our MoD system as firstly it is a direct generalization of the approach used for autonomous MoD systems and secondly it yields a linear optimization problem (Section \ref{sec:reb1}) that is easy to compute and scale to large systems. However, as discussed in Section \ref{sec:realcustomers}, such approach only approximately balances the system (even though the approximation is often remarkably good). We then introduce a modified approach in Section \ref{sec:nonlinear} that relies on nonlinear optimization, which does ensure fairness while maintaining system performance, but incurs a higher computational cost. Collectively, the open-loop control approaches of Section \ref{sec:rebalancing} are useful for analysis and design tasks such as system sizing (Section \ref{sec:fleetsizing}) and drive the development of closed-loop policies (Section \ref{sec:realTime}).

\section{Analysis and Design of MoD Systems} \label{sec:rebalancing}

\subsection{Approximate MoD rebalancing} \label{sec:reb1}
In this section we formulate a linear optimization approach to (approximately) rebalance a MoD system. Specifically, we would like to manipulate our control variables $\lambda_i^{\mathrm{del}}$, $\psi_i$, $\eta_{ij}$, and $\xi_{ij}$ such that $\gamma^{(1)}_i = \gamma^{(1)}_j$ for all $i,j \in S^{(1)}$ and $\gamma^{(2)}_i = \gamma^{(2)}_j$ for all $i,j \in S^{(2)}$. To minimize the cost of MoD service, we would like to simultaneously minimize the mean number of rebalancing vehicles on the road (minimize energy use and possibly congestion), given by $\sum_{i,j} T_{ij}\,\xi_{ij}\,\psi_i$, as well as the number of rebalancing drivers needed, given by $\sum_{i,j} T_{ij}\, (\xi_{ij}\, \psi_i + \eta_{ij}\, \lambda_i^{\mathrm{del}})$. We can state this multi-objective problem as follows:
\begin{quote}{\bf MoD Rebalancing Problem (MRP):} Given a MoD system modeled as 2 closed Jackson networks, solve
\begin{align}
&\underset{\lambda_i^{\mathrm{del}}, \psi_i, \eta_{ij}, \xi_{ij}}{\text{minimize}} && \sum_{i,j} T_{ij}\, \xi_{ij}\, \psi_i \text{ and } \sum_{i,j} T_{ij}\, (\xi_{ij}\, \psi_i + \eta_{ij}\, \lambda_i^{\mathrm{del}}) \notag \\
&\text{subject to} && \gamma^{(k)}_i = \gamma^{(k)}_j \;\;\;\; i,j \in S^{(k)}, k = 1,2  \label{problem1}  \\
& && \sum_j \eta_{ij} = 1, \;\;\;\; \sum_j \xi_{ij} = 1, \notag \\
& && \eta_{ij} \geq 0,\;\; \xi_{ij} \geq 0,\;\; \lambda_i^{\mathrm{del}} \geq 0,\;\; \psi_i \geq 0 \notag \\
& && \lambda_i^{\mathrm{del}} \eta_{ij} \leq \lambda_i p_{ij} \;\;\;\; i,j \in \{1, \ldots, N\} \notag
\end{align}
\end{quote}

The two objectives are indeed aligned, i.e., minimizing the second objective will minimize the first as well -- this can easily be shown by applying Lemma \ref{lemma:3} and \ref{lemma:4}. The last constraint in the MRP ensures that the customer-driven Jackson network always has non-negative rates. Remarkably, the MRP can be solved as two decoupled linear optimization problems with the same form as in \cite{SLS-MP-MS-ea:13} (which uses a deterministic, fluidic model). This result, stated in Theorem \ref{th:1}, constitutes the main contribution of this section. Its proof relies on several supporting lemmas which can be found in the appendix. 
%=======
%Note that the two objectives are indeed aligned (this can easily be shown by applying Lemma A.3 and A.4 in \cite{RZ-MP:15Arxiv} to the optimization problem in \cite{SLS-MP-MS-ea:13}). The last constraint in the MRP ensures that the customer-driven Jackson network always has non-negative rates. Remarkably, the MRP can be solved as two decoupled linear optimization problems with the same form as in \cite{SLS-MP-MS-ea:13} (which uses a deterministic, fluidic model). This result, stated in Theorem \ref{th:1}, constitutes the main contribution of this section. Its proof relies on several supporting lemmas which can be found in the appendix of \cite{RZ-MP:15Arxiv}. 
%>>>>>>> .r4681

\begin{theorem}[Solution to MRP] \label{th:1}
Consider the following two decoupled linear optimization problems
\begin{align}
&\underset{\beta_{ij}}{\emph{minimize}} && \sum_{i,j} T_{ij}\, \beta_{ij}  \label{problembeta}\\
&\emph{subject to} && \sum_{j \neq i} (\beta_{ij} - \beta_{ji}) = \lambda_i - \sum_{j \neq i}\lambda_j\, p_{ji} \notag\\
& && 0 \leq \beta_{ij} \leq \lambda_i\, p_{ij} \notag
\end{align}
\vspace{-1em}
\begin{align}
&\underset{\alpha_{ij}}{\emph{minimize}} && \sum_{i,j} T_{ij}\, \alpha_{ij}  \label{problemalpha} \\
&\emph{subject to} && \sum_{j \neq i} (\alpha_{ij} - \alpha_{ji}) = -\lambda_i + \sum_{j \neq i}\lambda_j\, p_{ji} \notag\\
& && 0 \leq \alpha_{ij} \notag
\end{align}
These problems are always feasible. Let $\beta^*_{ij}$ and $\alpha^*_{ij}$ be optimal solutions to problems \eqref{problembeta} and \eqref{problemalpha} respectively. By making the following substitutions
\begin{align*}
\lambda_i^{\mathrm{del}} &= \sum_{j \neq i} \beta_{ij}^*, \\
\psi_i &= \sum_{j \neq i} \alpha^*_{ij}, \\
\eta_{ij} &= 
\begin{cases}
	0 & \text{if } i = j, \\
	\beta^*_{ij}/\lambda_i^{\mathrm{del}} & \text{if } \lambda_i^{\mathrm{del}} > 0, i \neq j, \\
	1/(N-1) & \text{otherwise},
\end{cases} \\
\xi_{ij} &= 
\begin{cases}
	0 & \text{if } i = j, \\
	\alpha^*_{ij}/\psi_i & \text{if } \psi_i > 0, i \neq j, \\
	1/(N-1) & \text{otherwise},
\end{cases}
\end{align*}
one obtains the optimal solution to the MRP. 
\end{theorem} 
\begin{proof}
Problem \eqref{problemalpha} is an uncapacitated minimum cost flow problem and is always feasible. The upper bound constraint in Problem \eqref{problembeta} constitutes a standard condition for the existence of a feasible solution in a minimum cost flow problem \cite[p. 220]{SLS-MP-MS-ea:13, BK-JV:07}. The main task of the proof is showing that the constraints $\gamma^{(k)}_i = \gamma^{(k)}_j$ are equivalent to the constraints in \eqref{problembeta} and \eqref{problemalpha}, which is shown in Lemmas \ref{lemma:3} and \ref{lemma:4} using the Perron-Frobenius theorem \cite{CDM:01}.
%The result follows by making the substitutions for $\beta_{ij}$ in terms of $\lambda_i^{\mathrm{del}}$ and $\eta_{ij}$, and for $\alpha_{ij}$ in terms of $\psi_i$ and $\xi_{ij}$, as outlined in Lemmas \ref{lemma:3} and \ref{lemma:4}, respectively.
\end{proof}
This result allows us to compute the open-loop control very efficiently and can be applied to very large systems comprising hundreds of stations. We apply this technique in the next section to compute the availability of vehicles at each station and in Section \ref{sec:fleetsizing} to the problem of ``sizing" a MoD system (i.e., determining the optimal fleet size and number of drivers).

\subsection{Availability of vehicles for real passengers} \label{sec:realcustomers}
In general, the availability of vehicles at each station in the customer-driven system is different from the taxi system. The approach in the previous section calculates the availability of the two systems separately, but the availability of vehicles in the taxi system applies not only to real customers, but to virtual customers as well. To calculate the availability for all (real) passengers, we must consider both systems concurrently. First, we note that the total throughput of both real and virtual customers for both networks is given by 
\begin{equation*}
\Lambda^{\text{tot}}_i(m_v, m_d) = \Lambda^{(1)}_i(m_v-m_d) + \Lambda^{(2)}_i(m_d).
\end{equation*}
The throughput of only real passengers is given by
\begin{equation*}
\Lambda^{\text{pass}}_i(m_v,m_d) = \Lambda^{(1)}_i(m_v-m_d) + \frac{\lambda_i^{\mathrm{del}}}{\lambda_i^{\mathrm{del}} + \psi_i} \Lambda^{(2)}_i(m_d),
\end{equation*}
where the second term on the right hand side reflects the fraction of real passengers in the taxi network. Thus, the vehicle availability for real passengers is given by
\begin{equation*}
A^{\text{pass}}_i(m_v, m_d) = \frac{\Lambda^{\text{pass}}_i(m_v, m_d)}{\lambda_i}.
\end{equation*}
With some algebraic manipulations, $A^{\text{pass}}_i(m_v,m_d)$ can be rewritten as
\begin{equation}
A^{\text{pass}}_i(m_v, m_d) = A^{(1)}_i(m_v-m_d) q_i + A^{(2)}_i(m_d) (1-q_i).
\label{eq:realA}
\end{equation}
Since $q_i$ is in general not the same for all $i$, the availability of vehicles for real passengers will \emph{not} be the same for every station. 
%This is seen clearly in Figure \ref{fig:realA}, where the difference in availability between stations increase as the vehicle-to-driver ratio increases. 
\begin{figure}[htbp!]
	\centering \hspace{-1em}
		\subfigure[$m_v / m_d = 3$]{\label{fig:realA3} \includegraphics[width=0.165\textwidth]{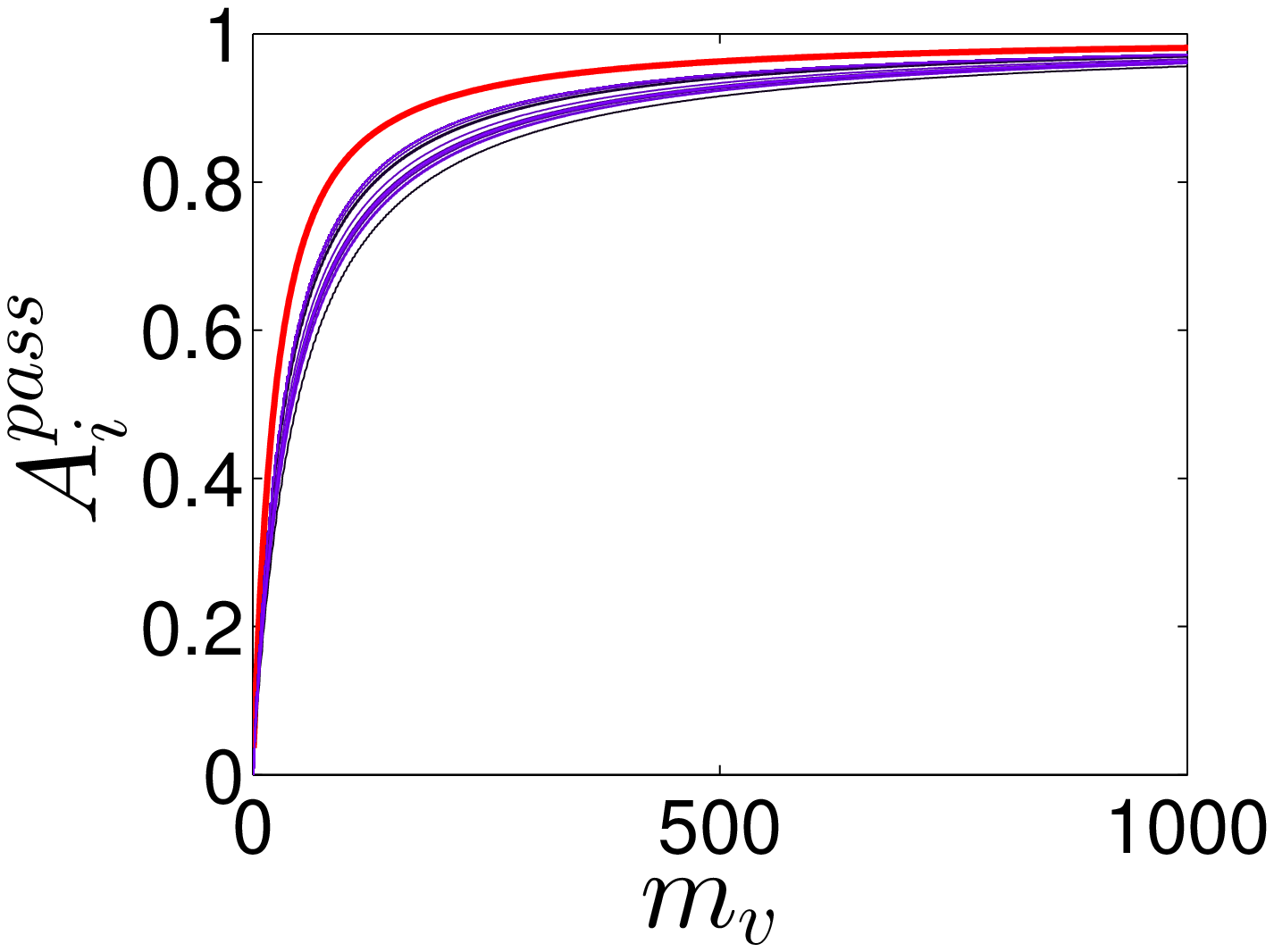}} \hspace{-1.2em}
		\subfigure[$m_v / m_d = 5$]{\label{fig:realA5} \includegraphics[width=0.165\textwidth]{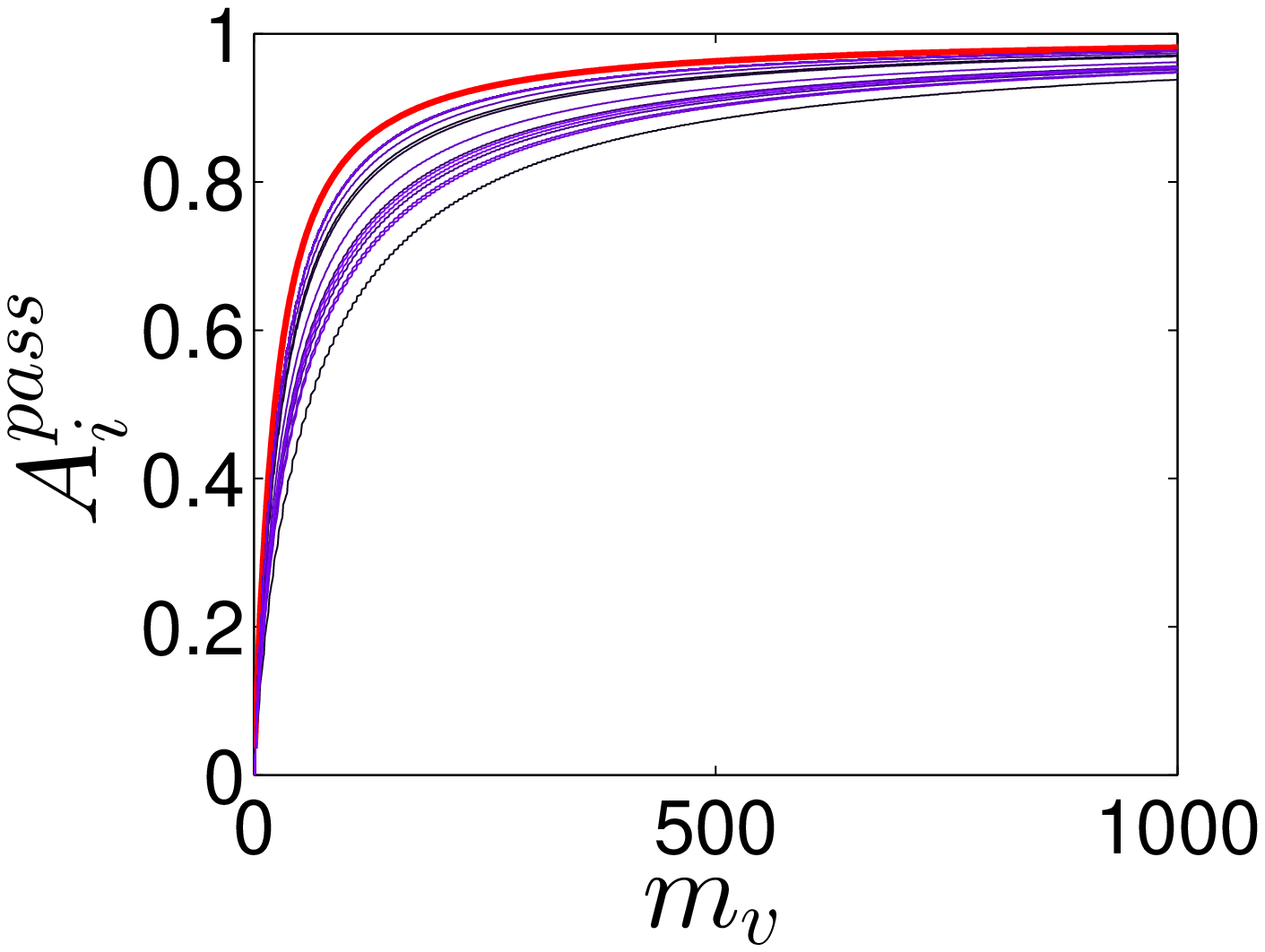}} \hspace{-0.9em}
		\subfigure[$m_v / m_d = 10$]{\label{fig:realA10}\includegraphics[width=0.165\textwidth]{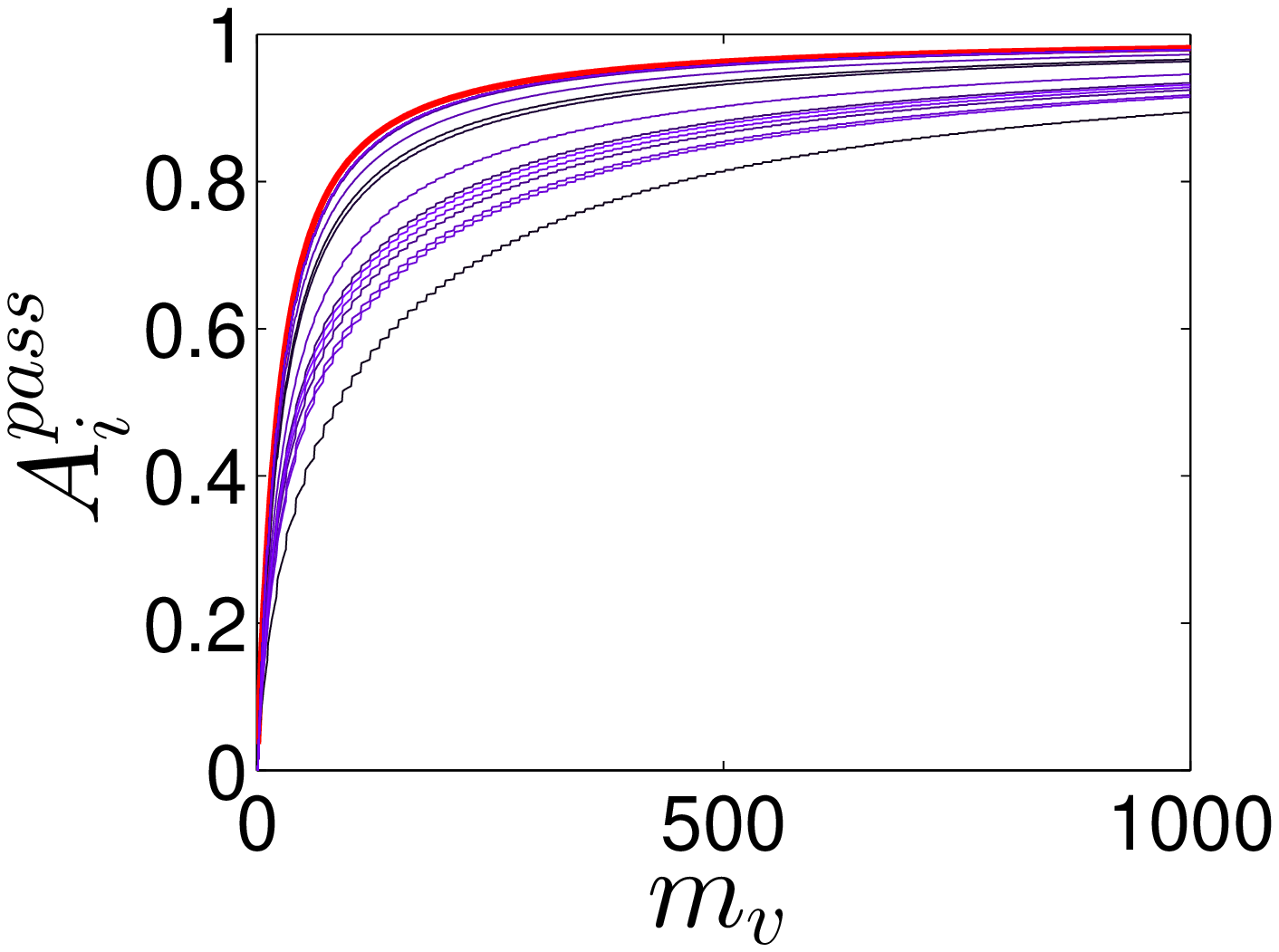}}			
	\caption{Overall vehicle availability for passengers for a randomly generated system with 20 stations. The red line shows the availability if there were as many drivers as vehicles (or an autonomous MoD system). \ref{fig:realA3} shows a vehicle-to-driver ratio of 3, \ref{fig:realA5} shows a vehicle-to-driver ratio of 5, and \ref{fig:realA10} shows a vehicle-to-driver ratio of 10.}
\vspace{-1.5em}
\label{fig:realA}
\end{figure}
Fig. \ref{fig:realA} shows that the rebalancing technique described in the previous section will produce unbalanced vehicle availabilities for real passengers. Furthermore, the degree of system imbalance grows with the vehicle-to-driver ratio, which intuitively makes sense since there are fewer drivers to rebalance the system when the vehicle-to-driver ratio is high. However, it is important to note that even though the availabilities at each station are not the same, as $m_v \rightarrow \infty$ and $m_d \rightarrow \infty$, the availabilities approach one for \emph{all} stations. 

The red line in Fig. \ref{fig:realA} shows the availability of the system if there were $m_v$ drivers and $m_v$ vehicles (or equivalently a taxi system or an autonomous MoD system). It is clear that the autonomous MoD system yields better performance both in terms of throughput (high availability) and fairness (same availability at all stations) due to the ability of every vehicle to perform rebalancing trips. This result presents a strong case for the advantages of autonomous MoD systems over current human-driven MoD systems in operation. 

%Figure \ref{fig:realA4} shows a system with a vehicle-to-driver ratio of 4, a number suggested by \ref{Smith2013ACC} as sufficient for a practical system. Indeed, a degree of imbalance exists between the stations, but the spread is much lower than the system with vehicle-to-driver ratio of 10 (Figure \ref{fig:realA10}). Indeed we see that if we size the system such that all the stations have over $90\%$ availability, the differences in availabilities between stations would be less than 5\%. This represents a degraded performance when compared to an autonomous MOD system (simply because not all vehicles can be rebalanced), but may still be practically acceptable if enough vehicles are put into the system to ensure that the availability of the worst station is ``high enough''. 

To validate these results, simulations are performed using a small 5-station system positioned in a $5 \times 5$ grid with vehicles traveling at a constant speed of 0.2 units per time step. In the simulation, customers arrive at each station according to a Poisson process with rate $\lambda_i$ and report their desired destinations. Based on the destinations, customers are delegated by a Bernoulli random variable to either drive themselves to their destination or be driven to their destination by a driver. If in either case a vehicle and/or driver is unavailable, the customer leaves the system. At each time step, after the customers are delegated, rebalancing is performed by generating ``virtual customers'' according to a Poisson process with rate $\psi_i$ and assigning the virtual customers to available drivers at each station. The availability for customers at each station is computed by dividing the number of successfully serviced customers by the total number of customer arrivals. Simulations are performed for 5 system sizes (number of vehicles and drivers) keeping the vehicle-to-driver ratio fixed. In order to capture the steady-state behavior of the system, each simulation is performed for 50,000 time steps. Fig. \ref{fig:lossSim} shows the simulated vehicle availabilities, averaged over 50 simulation runs, compared to the theoretical results computed using \eqref{eq:realA}. The simulations show remarkable consistency between the theoretical availability probabilities and experimental results, and that availability for real passengers indeed differ across stations under the control policy presented in Section \ref{sec:reb1}.
\begin{figure}[htbp]
	\centering 
	\includegraphics[width=0.35\textwidth]{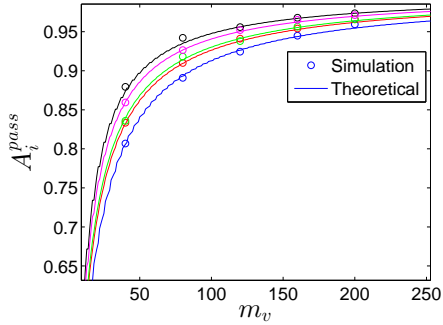}	 \vspace{-0.5em}	
	\caption{Validation of queueing model showing availability for real customers. A 5-station system is simulated with 40, 80, 120, 160, and 200 vehicles and a vehicle-to-driver ratio of 4. The circles represent mean availabilities over 50 simulation runs for each station. Each color represents a different station. The average standard deviation for the simulation results is 0.0189.}
\vspace{-1.7em}
\label{fig:lossSim}
\end{figure}

%Finally, we note that as we add more drivers to the system (lower the vehicle-to-driver ratio), the availability curves do not converge to the autonomous MOD curve with $m_v$ vehicles, but instead to the curve with $m_d$ vehicles (Figure \ref{fig:realA2}). This means that the system becomes used more and more as a taxi service and the empty vehicles are used less and less. 
%\vspace{-0.5em}
\subsection{Exact MoD rebalancing} \label{sec:nonlinear}
It is clear that applying the rebalancing constraints separately for the two networks as done in \eqref{problem1} does not yield a balanced system in terms of vehicle availability for all customers. Indeed, the constraints needed to balance availability for the passengers is
\begin{equation}
A^{\text{pass}}_i(m_v, m_d) = A^{\text{pass}}_j(m_v, m_d) \;\;\;\; \forall i,j \in \{1,...,N\}.
\label{eq:rebcon}
\end{equation} 
Note that this constraint is dependent on the number of vehicles and the number of drivers in the system, and thus cannot be reduced to a linear constraint in the decision variables. Taking into account the modified constraints, we reformulate our problem to the following:
\begin{quote}{\bf \mpmargin{Modified}{Stretch your imagination and come up with something better} MoD Rebalancing Problem (MMRP)}: Given a MoD system with $N$ stations, $m_v$ vehicles, and $m_d$ drivers modeled as two closed Jackson networks, solve
\begin{align}
&\underset{\lambda_i^{\mathrm{del}}, \psi_i, \eta_{ij}, \xi_{ij}}{\text{minimize}} && \sum_{i,j} T_{ij}\, \xi_{ij}\, \psi_i - c \sum_i A^{(2)}_i(m_v-m_d) \label{problem2} \\
&\text{subject to} && \gamma^{(1)}_i = \gamma^{(1)}_j  \notag  \\
& && A^{\text{pass}}_i(m_v, m_d) = A^{\text{pass}}_j(m_v, m_d) \notag \\
& && \sum_j \eta_{ij} = 1, \;\;\;\;\sum_j \xi_{ij} = 1 \notag \\
& && \eta_{ij} \geq 0,\; \xi_{ij} \geq 0,\; \lambda_i^{\mathrm{del}} \geq 0,\; \psi_i \geq 0 \notag \\
& && \lambda_i^{\mathrm{del}} \eta_{ij} \leq \lambda_i p_{ij} \;\;\;\; i,j \in \{1, \ldots, N\}. \notag
\end{align}
\end{quote}
The objective function now trades off two objectives that are not always aligned -- minimizing the number of rebalancing trips while maximizing the overall availability (note that the first constraint balances and maximizes the availability of the customer-driven system, so to maximize overall availability, we only need to maximize the availabilities in the taxi system). A weighting factor $c$ is used in this trade-off. The constraint $\gamma^{(1)}_i = \gamma^{(1)}_j$ is used in conjunction with \eqref{eq:rebcon} to ensure the availability of the customer-driven system remains balanced. The strategy is to use the taxi system to enforce the availability constraint for real customers with the intuition that the system operator has full control over the rebalancing of the taxi system while the rebalancing of the customer-driven system depends on the arrival process of the customers, which is subject to large stochastic fluctuations. If the customer-driven system becomes unbalanced, empty vehicles will accumulate at some stations for extended periods of time, decreasing the effective number of vehicles in the system (see Section \ref{sec:realTime}). 
%\mpmargin{The dependency on}{this sentence is a little bit obscure} $m_v$ and $m_d$ means that the rebalancing parameters in an optimal solution is only guaranteed to satisfy constraint \eqref{eq:rebcon} for a system consisting of $m_v$ vehicles and $m_d$ drivers. 

The modified availability constraint \eqref{eq:rebcon} is nonlinear and involves solving for $A^{(2)}_i$ using MVA at each iteration ($A^{(1)}_i$ is also needed, but only needs to be computed once). For systems of reasonably small size ($\sim20$ stations and $\sim1000$ vehicles), MVA can be carried out quickly ($<$ 1 sec). For larger networks, an approximate MVA technique exists which involves solving a set of nonlinear equations rather than iterating through all values of $m$ \cite{MR-SSL:80}. The MMRP can be solved using nonlinear optimization techniques for a given number of vehicles and drivers. We let $A^*$ represent the balanced availability $A^{\text{pass}}_i$ obtained by solving the MMRP. 
\vspace{-1.5em}
\begin{figure}[htbp]
\centering
\subfigure[$c = 1, A^* = 0.5327$]{\label{fig:nonlinC1} \includegraphics[width=0.23\textwidth]{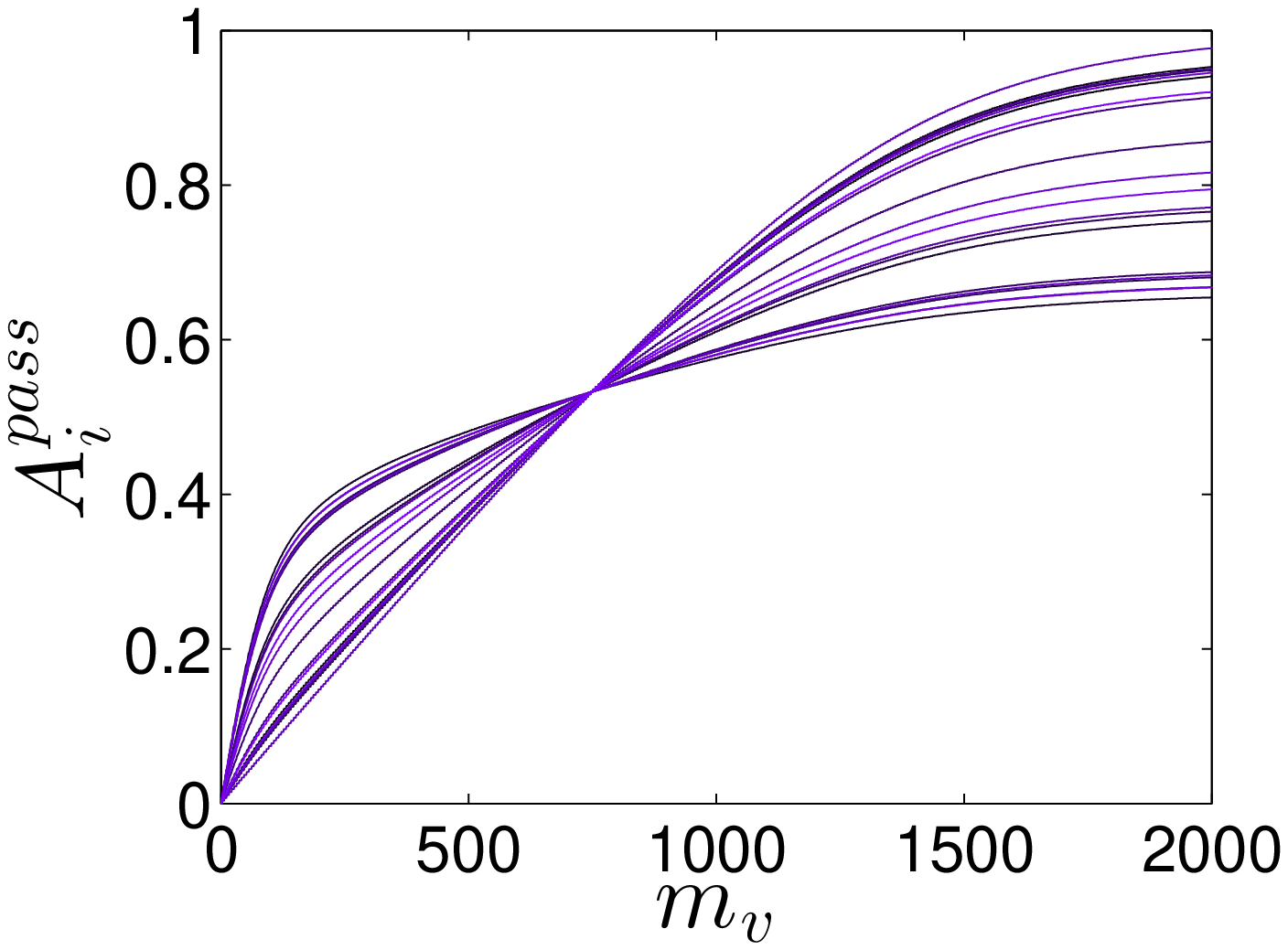}} \hspace{-1em}
\subfigure[$c = 10, A^* = 0.8988$]{\label{fig:nonlinC10} \includegraphics[width=0.23\textwidth]{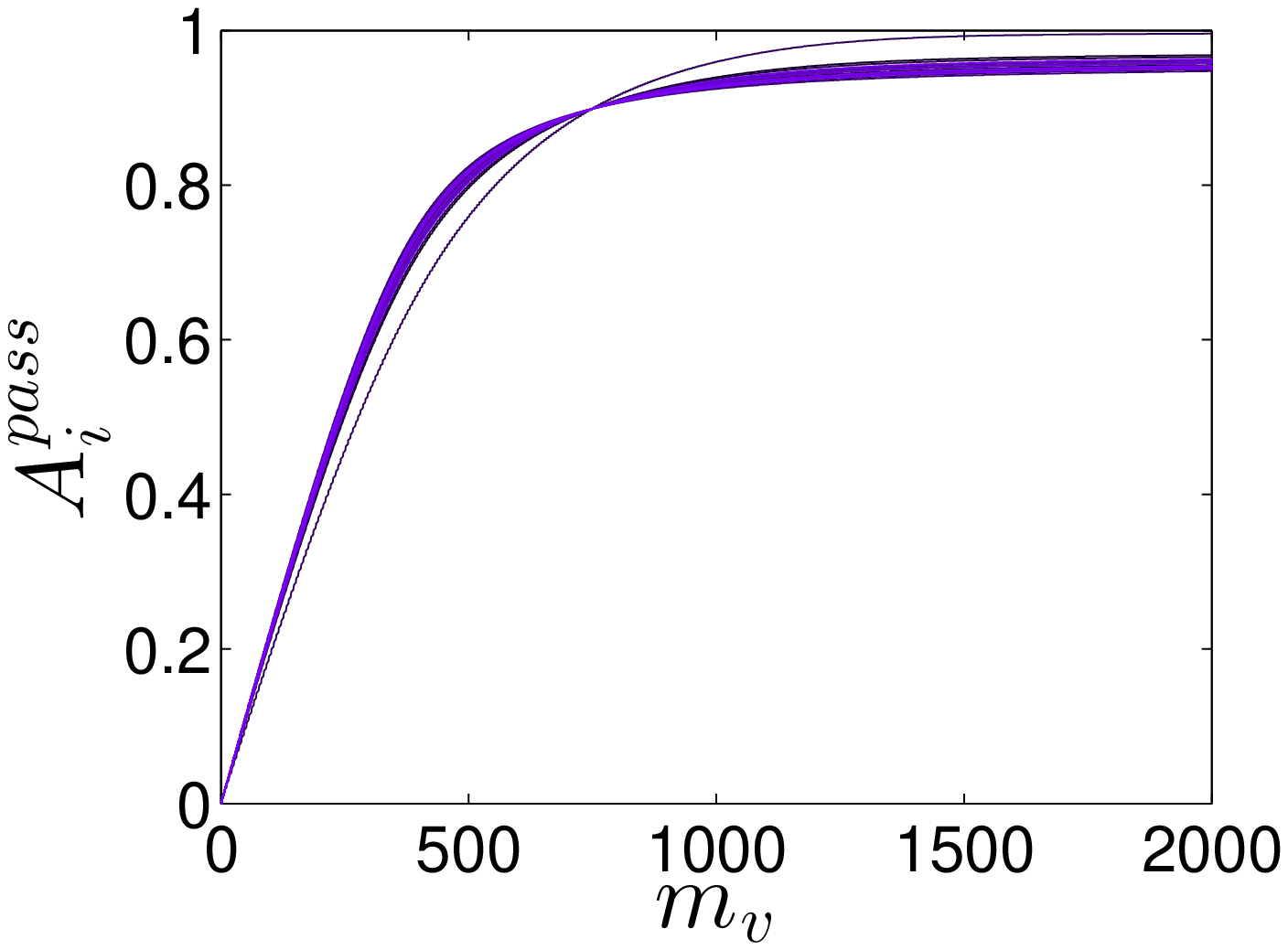}} \\
\vspace{-1em}
\subfigure[Pareto optimal curve of MMRP ($c = 1,2,3,4,5,6,10,15,20,50$)]{\label{fig:nonlinPareto} \includegraphics[width=0.23\textwidth]{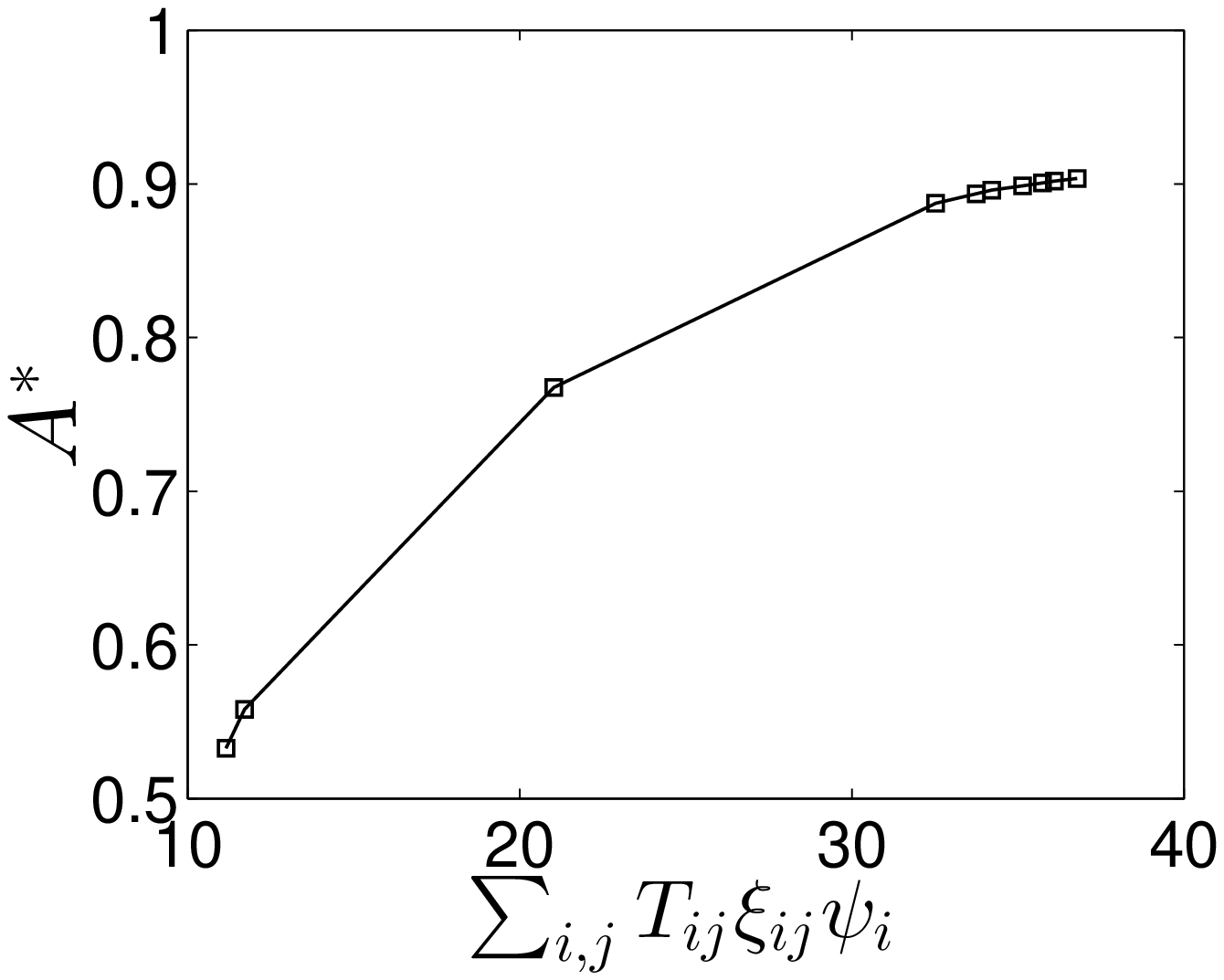}}
\hspace{-1em}
\subfigure[Linear solution. ]{\label{fig:linfrac5NY} \includegraphics[width=0.23\textwidth]{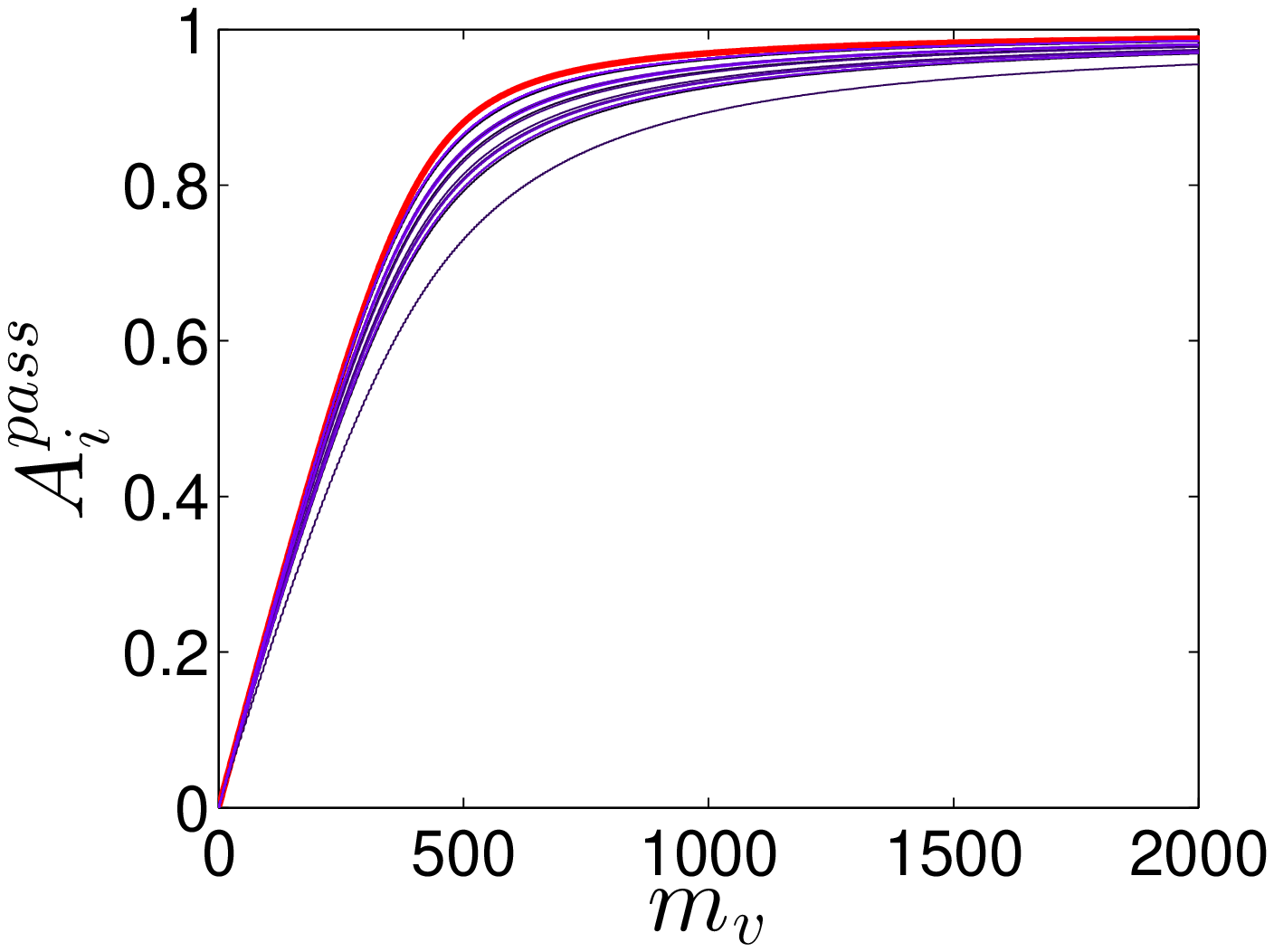}}
\caption{Nonlinear optimization results for a 20 station system based on Lower Manhattan taxi trip data. \ref{fig:nonlinC1} shows the optimized availability curves for $c = 1$. \ref{fig:nonlinC10} shows the optimized availability curves for $c = 10$. \ref{fig:nonlinPareto} shows the Pareto optimal curve obtained by increasing $c$ from 1 to 50. The x-axis can be interpreted as the average number of rebalancing vehicles on the road. \ref{fig:linfrac5NY} shows the linear optimization results for comparison.}
\vspace{-0.5em}
\label{fig:nonlinear}
\end{figure}

To demonstrate this technique on a realistic system, key system parameters (arrival rates, routing probabilities, and travel times) were extracted from a portion of a data set of New York City taxi trips\footnote{Courtesy of the New York City Taxi \& Limousine Commission.}. Specifically, a 20-station system was created using taxi trips within Lower Manhattan (south of 14th St.) between 10 and 11am on March 1, 2012. The MMRP is solved for this system with 750 vehicles and 150 drivers ($m_v/m_d=5$). Fig. \ref{fig:nonlinear} shows the resulting availability curves and the trade-off between rebalancing rate and system performance. 

Fig. \ref{fig:nonlinear} shows that as the weighting factor $c$ is increased, vehicle availability increases at a cost of an increased number of rebalancing trips up to a point and levels off (in this case around 90\%). This result compares favorably with the linear solution (\ref{fig:linfrac5NY}), where at $m_v = 750$, the availabilities range from 0.84 to 0.94. 
In general, the linear optimization technique appears suitable for computing a first approximation of key design parameters of the system, and the nonlinear technique can be used to further refine the solution.
%In terms of the number of rebalancing trips, the linear solution yielded $\sum_{i,j} T_{ij} \xi_{ij} \psi_i = 32.4$, which is only \mpmargin{slightly lower than the nonlinear result}{Can we summarize these insights, e.g., saying, linear optimization appear to be adequate from most practical purposes?} of $\sim35$. 
Finally, compared to an autonomous MoD system with the same number of vehicles (red line in Fig. \ref{fig:linfrac5NY}), the overall availability is 5\% lower (90\% vs. 95\%). This further shows that autonomous MoD systems would achieve higher levels of performance compared to MoD systems. 

\subsection{Application to system sizing} \label{sec:fleetsizing}
% 3 thresholds - 85%, 90%, 95%
% 3 Manhattan areas to account for different usage patterns (Lower Manhattan, Midtown, Central Park)
Though the linear programming approach (Section \ref{sec:reb1}) does not yield identical availabilities across all stations, it is nonetheless useful for applications such as fleet sizing due to its scalability and efficiency. In this section we provide a simplified example of how to use the MRP approach to gain insight into the optimal vehicle-to-driver ratio ($m_v/m_d$) of a MoD system. The idea is to find the optimal number of vehicles and drivers that would minimize total cost (or maximize profit) while maintaining an acceptable quality of service. For this simple example, the total cost (normalized by the cost of a vehicle) is
\begin{equation}
\vspace{-0.1em}
c_{\mathrm{total}} = m_v + c_r m_d,
\label{eq:cost}
\vspace{-0.1em}
\end{equation}
where $c_r$ is the cost ratio between a vehicle and a driver. It is reasonable to assume that the cost of a driver is greater than the cost of a vehicle, so $c_r >=1$. Three MoD systems are generated using portions of the New York City taxi data: 1) Lower Manhattan (A1), 2) Midtown Manhattan (A2), and 3) Upper Manhattan (A3). Taxi trips within each region are aggregated and clustered into 20 stations, and the system parameters ($\lambda_i$, $p_{ij}$, and $T_{ij}$) are estimated. Different travel patterns in the three systems allow us to generalize our insights about the optimal $m_v/m_d$ required to minimize cost. For each system with a fixed $m_v/m_d$, the MRP is solved and the number of vehicles and drivers needed are found such that the lowest availability across all the stations is greater than the availability threshold. Three availability thresholds are investigated (85\%, 90\%, and 95\%). Fig. \ref{fig:fleet1} shows the total cost as it varies with the vehicle-to-driver ratio and with $c_r$ for Lower Manhattan with 90\% availability threshold. The optimal vehicle-to-driver ratio is the minimum point of each line in \ref{fig:fleet1}. Fig. \ref{fig:fleet2} shows the optimal vehicle-to-driver ratios plotted against the cost ratio $c_r$ for all three Manhattan suburbs and all three availability thresholds. 
\begin{figure}[htbp]
\vspace{-1.5em}
\centering
\subfigure[Lower Manhattan (A1)-90\%]{\label{fig:fleet1} \includegraphics[width=0.23\textwidth]{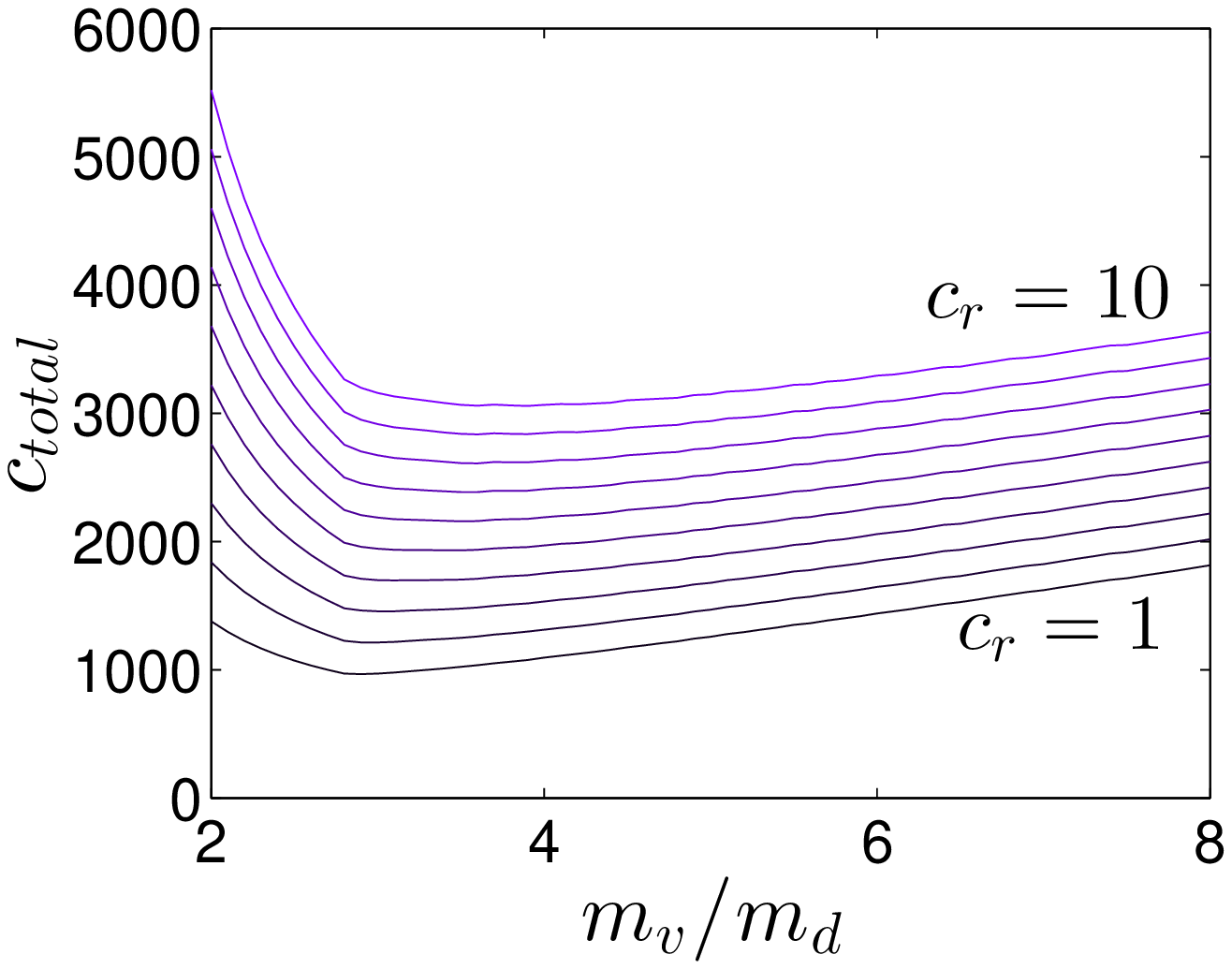}} \hspace{-1em}
\subfigure[Optimal vehicle-to-driver ratio]{\label{fig:fleet2} \includegraphics[width=0.24\textwidth]{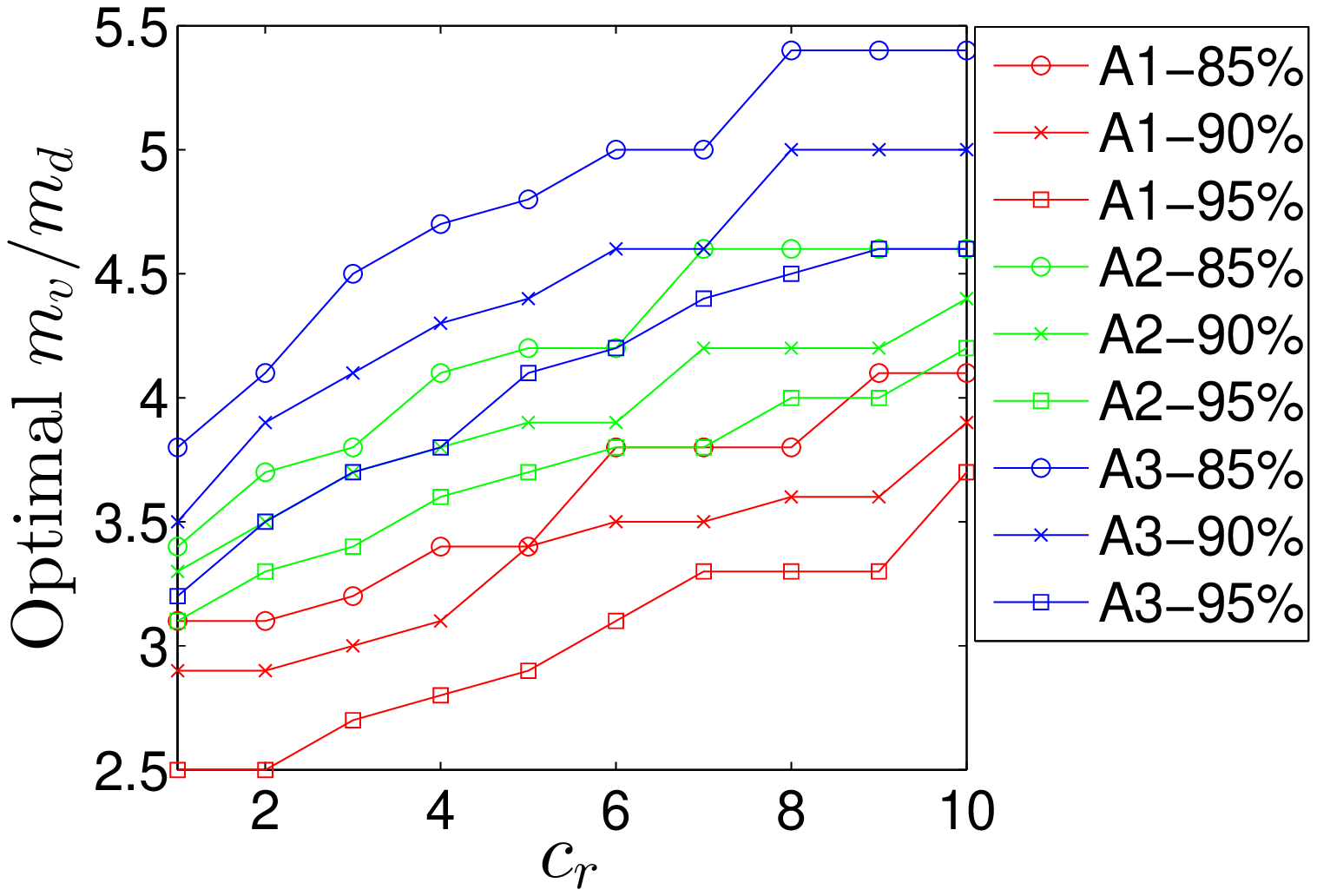}}
\caption{\ref{fig:fleet1} shows the total cost as a function of the vehicle-to-driver ratio for $c_r$ values ranging from 1 to 10. $m_v$ and $m_d$ values at each point in each curve can be solved using \eqref{eq:cost}. $m_v$ and $m_d$ satisfies the constraint that the availability at each station is greater than the threshold of 90\%. \ref{fig:fleet2} shows the optimal vehicle-to-driver ratio for the 3 suburbs of Manhattan and 3 availability thresholds (85\%, 90\%, 95\%). The curve A1-90\% in \ref{fig:fleet2} is constructed from the minimum points of each curve in \ref{fig:fleet1}.}
\vspace{-0.7em}
\label{fig:fleetSizing}
\end{figure}

A few insights can be gained from this example. First, the optimal $m_v/m_d$ ratio does not significantly increase with increasing cost ratio. % (e.g., the optimal $m_v/m_d$ increases by one as $c_r$ increases by ten). 
Second, the optimal $m_v/m_d$ ratio decreases as the availability threshold is raised, consistent with the idea that a high quality of service requires more rebalancing, and thus more drivers. Third, the optimal $m_v/m_d$ ratio is clearly different for each of the Manhattan suburbs (which highlights the important system-dependent nature of this value) but stays between 3 and 5 for a wide range of cost ratios. This example shows the applicability of the queueing network approach to the design and analysis of MoD systems. Similar studies can be done with the nonlinear approach, which will yield higher predictive fidelity but at a higher computation cost. 
\vspace{-0.2em}
\section{Closed-loop Control of MoD Systems} \label{sec:realTime}
\vspace{-0.2em}
In this section we formulate a real-time closed-loop control policy by drawing inspiration from the open-loop counterparts in Section \ref{sec:rebalancing}. Our closed-loop policy relies on receding horizon optimization and  is targeted towards a practical scenario where customers would wait in line for the next available vehicle rather than leave the system. The control policy must perform two tasks: 1) rebalance vehicles throughout the network by issuing instructions to drivers, and 2) assign vehicles (with or without driver) to new customers at each station. For simplicity, as in Section \ref{sec:rebalancing}, we perform these tasks \emph{separately} by designing  a vehicle rebalancing policy and a customer-assignment policy. A vehicle rebalancing policy was introduced in \cite{MP-SLS-EF-DR:12} for autonomous MoD systems, which has been shown to be quite effective \cite{RZ-MP:14}, hence we adapt it for our system with little modification. The customer-assignment policy is trickier, and we propose a mixed-integer linear program (MILP) to select the best assignment based on the current state of the system. The proposed policy enforces the following operation scenario for the MoD system: Customers arriving at each station join a queue of ``unassigned'' customers. A system-wide optimization problem is solved to try to assign as many customers as possible while keeping the customer-driven vehicles balanced. Once a customer is assigned, he/she moves to the departure queue where he/she will depart with an empty vehicle or with a taxi. The optimization procedure is performed every time a departure queue is empty and there are unassigned customers. The notion of keeping the customer-driven vehicles balanced at each station stems from early studies we performed using simple heuristic policies, where we observed customer-driven vehicles aggregating at a small number of stations unused for long periods of time, effectively decreasing the number of vehicles in the system. This observation inspired the formulation of the MMRP (Section \ref{sec:nonlinear}) as well as the real-time policy. 

Let $n^v_{ij}$ be the number of customers traveling from station $i$ to $j$ to be assigned to drive themselves. Let $n^d_{ij}$ be the number of customers traveling from station $i$ to $j$ to be assigned to a taxi. Denote by $v^e_i$ the number of excess unassigned customer-driven vehicles at station $i$, $v^{t}_{ji}$ the number of customer-driven vehicles traveling from station $j$ to $i$, and $v^{a}_{ji}$ the number of customer-driven vehicles at station $j$ assigned to travel to station $i$ but that have not yet left the station. Assuming these quantities are known, the number of customer-driven vehicles at a future time step is $v^{+}_i = v^e_i + \sum_j (v^{a}_{ji} + v^{t}_{ji} + n^v_{ji} - n^v_{ij})$. We can define a desired vehicle distribution to be, for example, $v^{des}_i = (m_v - m_d) \lambda_i/\sum_i \lambda_i$. The assignment policy is given by solving the following optimization problem
\begin{align}
&\underset{n^d_{ij}, n^v_{ij}}{\text{minimize}} && \sum_i | v^{+}_i - v^{des}_i | - w \sum_{i,j} (n^d_{ij} + n^v_{ij}) \label{milp} \\
&\text{subject to} && n^d_{ij} + n^v_{ij} \leq c^u_{ij}  \notag  \\
& && \sum_j n^v_{ij} \leq v^e_i, \;\;\;\; \sum_j n^d_{ij} \leq d^u_i \notag \\
& && n^v_{ij} \geq 0, \; n^d_{ij} \geq 0, \; n^v_{ij} \in \mathbb{Z}, \; n^d_{ij} \in \mathbb{Z}, \notag
\end{align}
where $c^u_{ij}$ is the number of unassigned customers traveling from $i$ to $j$, $d^u_i$ is the number of unassigned drivers at station $i$, and $w$ is a weighting factor. The objective function trades off the relative importance of system balance and customer wait times (increasing $w$ would allow the system to assign more customers and reduce wait times). The constraints ensure that the assignment policy is feasible (there are enough vehicles, drivers, and customers). Problem \eqref{milp} is formulated as a MILP and solved using the IBM CPLEX solver \cite{CPLEX}. 

To assess the performance and stability of the real-time policy, simulations were performed using a 20-station system based on the travel patterns of Lower Manhattan (as in Section \ref{sec:fleetsizing}). The simulations were performed with 728 vehicles and 243 drivers ($m_v/m_d = 3$), which is the minimum size to reach 90\% availability across all stations, according to the analysis in Section \ref{sec:fleetsizing}. Each simulation was performed for 3 hours with a time step of 2 seconds. To reduce the effects of initial conditions, data was only collected for the final 2 hours. Twenty simulations were performed and the customer wait times at each time step are collected. Fig. \ref{fig:waittimes} shows the average customer wait times ($\pm 1$ standard deviation) for the station with the longest wait times. It is interesting to note that the station with the longest wait times in simulation is also the station with the lowest availability in the queueing analysis. From Fig. \ref{fig:waittimes} we see that the real-time policy yields a stable system (in terms of customer wait times) and that 90\% availability corresponds in this case to a reasonable average wait time of 7 minutes. 
\begin{figure}[htbp]
	\centering \vspace{-1em}
	\includegraphics[width=0.4\textwidth]{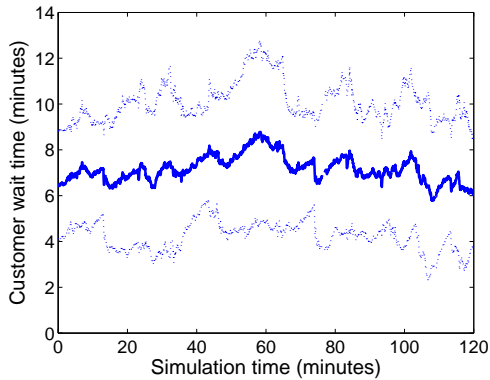}		
	\caption{Average customer wait time for the ``worst'' station. Data averaged over 20 simulations. Dotted lines indicate $\pm 1$ standard deviation.}
\vspace{-1em}
\label{fig:waittimes}
\end{figure}

In our simulations, the driver assignment problem \eqref{milp} with 820 variables was solved in 5ms on average. Since the constraints consist mostly of bounding hyperplanes, the feasible set is easy to compute and the problem should scale well to large systems, especially since it only needs to be solved once every few minutes. The results in this section showed that high availability as predicted by the open-loop control correspond well to real system performance under closed-loop control. 
\section{Conclusions and Future Work} \label{sec:conc}
In this paper we presented a queueing network model of a MoD system and developed two open-loop control approaches useful for  design tasks such as system sizing. We applied such approaches to a system sizing example for three Manhattan neighborhoods, which showed that the optimal  vehicle-to-driver ratio  is between  3 and 5. Drawing insights from these techniques, we developed a closed-loop real-time control policy  and demonstrated its effectiveness on real traffic data.

This work paves the way for several important extensions. First, we plan to include other methods of rebalancing drivers such as allowing them to use public transit or to shuttle multiple other drivers to stations with excess unused vehicles.  Second, we would like to study the effects of time-varying demand on the system, and how this impacts the amount of rebalancing that is required (for example, a small team of drivers may be able to balance the entire fleet at night or during off-peak hours). Third, it is of interest to include congestion effects in our model. A possible strategy is to modify the IS nodes by considering a finite number of servers, representing the capacity of the road. Fourth, we plan to test our strategies on microscopic and mesoscopic models of transportation networks.  Finally, we would like to incorporate the effects of dynamic pricing incentives for customers on the amount of rebalancing that is required. 

\section*{Acknowledgments} The authors would like to acknowledge insightful discussions on this topic with S. Beiker, E. Frazzoli, D. Rus, M. Schwager, and S.  Smith.
%An economic analysis is needed to determine if this type of system offers cost benefits for customers compared to taxi systems and autonomous MoD systems. 
\bibliographystyle{IEEEtran}
\bibliography{../../../bib/alias,../../../bib/main,arxivbib}
\pagebreak
\appendix
\renewcommand{\thesection}{\Alph{section}}
In this section we provide the supporting lemmas necessary for the proof of Theorem \ref{th:1}. We begin by restating two lemmas introduced in \cite{RZ-MP:14} for completeness. The first lemma allows the balance equations of the Jackson network to be solved by considering only the SS nodes.
\begin{lemma}[Folding of balance equations]
Consider either System 1 or System 2 from Section \ref{sec:jmodel}. The relative throughputs $\pi$'s for the SS nodes can be found by solving the reduced balance equations
\begin{equation}
\pi_i^{(k)} = \sum_{j \in S^{(k)}} \pi^{(k)}_j p_{ji}^{(k)} \;\;\;\; \forall i \in S^{(k)}, \, k = \{1,2\},
\end{equation}
where SS nodes are considered in isolation. The $\pi$'s for the IS nodes are then given by
\begin{equation}
\pi_i^{(k)} = \pi^{(k)}_{\text{Parent}(i)} p^{(k)}_{\text{Parent}(i)\text{Child}(i)} \;\; \forall i \in I^{(k)}, k = \{1,2\}.
\end{equation}
\end{lemma}

\begin{lemma}\label{lemma:2}
For any rebalancing policy $\{\psi_i\}_i$ and $\{\xi_{ij}\}_{ij}$, it holds for all $i \in S^{(2)}$
\begin{enumerate}
\item $\gamma^{(2)}_i > 0$, 
\item $(\lambda_i^{\mathrm{del}} + \psi_i)\, \gamma^{(2)}_i = \sum_{j \in S^{(2)}} \gamma^{(2)}_j\, (\psi_j\, \xi_{ji} + \lambda_j^{\mathrm{del}}\, \eta_{ji})$.
\end{enumerate}
Similarly, for System 1,
\begin{enumerate}
\item $\gamma^{(1)}_i > 0$, 
\item $(\lambda_i - \lambda_i^{\mathrm{del}})\, \gamma^{(1)}_i = \sum_{j \in S^{(1)}} \gamma^{(1)}_j\, (\lambda_j\, p_{ji} - \lambda_j^{\mathrm{del}}\, \eta_{ji})$.
\end{enumerate}
\end{lemma}

In the next two lemmas, we introduce new optimization variables $\{\alpha_{ij}\}_{ij}$ and $\{\beta_{ij}\}_{ij}$ and show that the constraints $\gamma_i = \gamma_j$ in the MRP are equivalent to linear constraints in these new variables. The proofs are similar to the proof of Theorem IV.3 in \cite{RZ-MP:14}. 
\begin{lemma}[Constraint equivalence for System 1] \label{lemma:3}
Assume that $\beta_{ij}$ is given. Set $\lambda_i^{\mathrm{del}} = \sum_{j \neq i} \beta_{ij}$, $\eta_{ii} = 0$, and for $j \neq i$,
\begin{equation*}
\eta_{ij} = 
\begin{cases}
	\beta_{ij}/\lambda_i^{\mathrm{del}} & \text{if } \lambda_i^{\mathrm{del}} > 0, \\
	1/(N-1) & \text{otherwise}.
\end{cases}
\end{equation*}
With this definition, the constraint
\begin{equation}
\sum_{j \in S^{(1)}, j \neq i} (\beta_{ij} - \beta_{ji}) = \lambda_i - \sum_{j \in S^{(1)}, j \neq i}\lambda_j\, p_{ji}
\label{eq:constraintb}
\end{equation}
is equivalent to the constraint
\begin{equation*}
\gamma^{(1)}_i = \gamma^{(1)}_j, \;\;\;\; i, j \in S^{(1)}.
\end{equation*}
\end{lemma}
\begin{proof}
First, rewrite \eqref{eq:constraintb} in terms of $\lambda_i^{\mathrm{del}}$ and $\eta_{ij}$. We then have
\begin{equation*}
\lambda_i - \lambda_i^{\mathrm{del}} = \sum_{j \neq i} (\lambda_j\, p_{ji} - \lambda_j^{\mathrm{del}}\, \eta_{ji}).
\end{equation*}
Substituting this expression into the last statement of Lemma \ref{lemma:2}, we have
\begin{equation}
\left(\sum_{j \neq i} (\lambda_j\, p_{ji} - \lambda_j^{\mathrm{del}}\, \eta_{ji})\right) \gamma^{(1)}_i = \sum_{j \neq i} \gamma^{(1)}_j\, (\lambda_j\, p_{ji} - \lambda_j^{\mathrm{del}}\, \eta_{ji}).
\label{eq:markov}
\end{equation} 
Let $\varphi_{ij}:= \lambda_j\, p_{ji} - \lambda_j^{\mathrm{del}}\, \eta_{ji}$ and $\zeta_{ij}:= \varphi_{ij}/\sum_{j} \varphi_{ij}$. Note that $\sum_j \varphi_{ij} = \lambda_i - \lambda_i^{\mathrm{del}} = \lambda^{(1)}_i > 0$ by assumption. The variables $\zeta_{ij}$ can be considered transition probabilities of an irreducible Markov chain, and \eqref{eq:markov} can be rewritten in matrix form as $Z\gamma^{(1)} = \gamma^{(1)}$.
%where $\gamma^{(1)} = (\gamma_1^{(1)}, ..., \gamma_N^{(1)})^T$ and $Z$ is an irreducible row stochastic matrix with 
Matrix $Z$ is an irreducible, row stochastic matrix, so by the Perron-Frobenius theorem \cite{CDM:01}, the eigenspace associated with the eigenvalue 1 is one-dimensional. Therefore the unique solution to $Z \gamma^{(1)} = \gamma^{(1)}$ (up to a scaling factor) is the vector $(1, ..., 1)^T$, so $\gamma^{(1)}_i = \gamma^{(1)}_j$ for all $i, j$. 
\end{proof}
\begin{lemma}[Constraint equivalence for System 2] \label{lemma:4}
Assume that $\alpha_{ij}$ is given. Set $\psi_i = \sum_{j \neq i} \alpha_{ij}$, $\xi_{ii} = 0$, and for $j \neq i$,
\begin{equation*}
\xi_{ij} = 
\begin{cases}
	\alpha_{ij}/\psi_i & \text{if } \psi_i > 0, \\
	1/(N-1) & \text{otherwise}.
\end{cases}
\end{equation*}
With this definition, the constraint
\begin{equation}
\sum_{j \neq i} (\alpha_{ij} - \alpha_{ji}) = \sum_{ j \neq i} (\beta_{ji} - \beta_{ij})
\label{eq:constrainta}
\end{equation}
is equivalent to the constraint
\begin{equation*}
\gamma^{(2)}_i = \gamma^{(2)}_j, \;\;\;\; i, j \in S^{(2)}.
\end{equation*}
\end{lemma}
The proof is essentially identical to the proof of Lemma \ref{lemma:3} and is omitted. Furthermore, we can substitute \eqref{eq:constraintb} into \eqref{eq:constrainta} and rewrite \eqref{eq:constrainta} as
\begin{equation}
\sum_{j \neq i} (\alpha_{ij} - \alpha_{ji}) = -\lambda_i + \sum_{j \neq i} \lambda_j p_{ji}.
\end{equation}
With this substitution, we have decoupled the original MRP constraints to those associated with System 1 ($\lambda_i^{\mathrm{del}}$ and $\eta_{ij}$) and those associated with System 2 ($\psi_i$ and $\xi_{ij}$).

\end{document}

%% file: preamble.tex
\usepackage{color}
\usepackage{amsmath}
\usepackage{amssymb}
\usepackage{graphicx}
\usepackage{comment,xspace}
\usepackage{hyperref}
\usepackage{fancybox}

\newtheorem{theorem}{Theorem}[section]

\newtheorem{lemma}[theorem]{Lemma}